\documentclass{amsart}
\usepackage{amssymb}
\usepackage{verbatim}
\usepackage{mathrsfs}
\usepackage{dsfont}

\newtheorem{theorem}{Theorem}[section]
\newtheorem{lemma}[theorem]{Lemma}
\newtheorem{corollary}[theorem]{Corollary}
\newtheorem{proposition}[theorem]{Proposition}
\newtheorem{observation}[theorem]{Observation}
\theoremstyle{definition}
\newtheorem{definition}[theorem]{Definition}
\newtheorem{example}[theorem]{Example}

\theoremstyle{remark}
\newtheorem{remark}[theorem]{Remark}

\newcommand{\I}{{\mathds {1}}}

\newcommand{\cB}{{\mathcal B}}

\newcommand{\cE}{{\mathcal E}}

\newcommand{\cF}{{\mathcal F}}
\newcommand{\cH}{{\mathcal H}}

\newcommand{\cM}{{\mathcal M}}

\newcommand{\cP}{{\mathcal P}}

\newcommand{\Rn}{{\rm I\!R}} 
\newcommand{\Nn}{{\rm I\!N}} 
\newcommand{\Cn}{{\setbox0=\hbox{
$\displaystyle\rm C$}\hbox{\hbox
to0pt{\kern0.6\wd0\vrule height0.9\ht0\hss}\box0}}} 

\numberwithin{equation}{section}

\newcommand{\tT}{\widetilde{T}}
\newcommand{\tx}{\widetilde{x}}
\newcommand{\jed}{{\mathbb{I}}}

\begin{document}

\title{Quantum dynamics on Orlicz spaces}

\author{L. E. Labuschagne}

\address{DST-NRF CoE in Math. and Stat. Sci,\\ Unit for BMI,\\ Internal Box 209, School of Comp., Stat., $\&$ Math. Sci.\\
NWU, PVT. BAG X6001, 2520 Potchefstroom\\ South Africa}
\email{Louis.Labuschagne@nwu.ac.za}

\author{W. A. Majewski}

\address{Institute of Theoretical Physics and Astrophysics, The Gdansk University, Wita Stwosza 57,\\
Gdansk, 80-952, Poland and Unit for BMI, North-West-University, Potchefstroom, South Africa}
\email{fizwam@univ.gda.pl}

\date{\today}
\subjclass[2010]{46L55, 47L90 (Primary); 46L51, 46L52, 46E30, 81S99, 82C10 (Secondary)}
\thanks{The authors would like to thank Adam Skalski for his willingness to share unpublished work on non-KMS-symmetric Markov operators. The contribution of L. E. Labuschagne is based on research partially supported by the National Research Foundation (IPRR Grant 96128). Any opinion, findings and conclusions or recommendations expressed in this material, are those of the author, and therefore the NRF do not accept any liability in regard thereto. For W. A. Majewski the partial support of the Foundation for Polish Science TEAM project cofinanced by the EU European Regional Development Fund
is acknowledged.}

\begin{abstract}
Quantum dynamical maps are defined and studied for quantum statistical physics based on Orlicz spaces. This complements earlier work  \cite{ML1} where we made a strong case for the assertion that statistical physics of regular systems should properly be based on the pair of Orlicz spaces $\langle L^{\cosh - 1}, L\log(L+1)\rangle$. The present paper therefore in some sense ``completes'' the picture by showing that even in the most general non-commutative contexts, completely positive Markov maps satisfying a natural Detailed Balance condition, canonically admit an action on a large class of quantum Orlicz spaces. This is achieved by the development of a new interpolation technique, specifically suited to the above context, for extending the action of such maps to the appropriate intermediate spaces of the pair $\langle L^\infty,L^1\rangle$. Moreover, it is shown that quantum dynamics in the form of Markov semigroups described by some Dirichlet forms naturally extends to the context proposed in \cite{ML1}.
\end{abstract}

\maketitle

\tableofcontents

\section{Introduction}
In our previous paper \cite{ML1} we have argued that statistical physics of regular systems, both classical and quantum, should be based on the pair of Orlicz spaces $\langle L^{\cosh - 1}, L\log(L+1)\rangle$. 

We remind the reader that a regular observable is characterized by the property of finiteness of all its moments. Although at first sight this property can be regarded as restrictive, the important point to note here is that the standard formulation of classical (quantum) theory based on the pair of Banach spaces $\langle L^{\infty}, L^1\rangle$ ( $\langle B(\cH), \mathfrak{F}(\cH)\rangle$, respectively) with pure deterministic time evolution, as well as such important classes of L\'evy processes as Wiener and Poisson processes do satisfy the above regularity requirements. Moreover, the proposed extension of the allowed family of observables includes regular \textit{unbounded} observables.

However, to get a fully-fledged theory, a description of dynamics should be provided. In particular, one wants to describe dynamical semigroups within the proposed scheme based on the pair of Orlicz spaces $\langle L^{\cosh - 1}, L\log(L+1)\rangle$. The main difficulty in carrying out such a description is that the standard interpolation theory must be adapted to a framework based on non-commutative Orlicz spaces.
We will do this in sections devoted to the study of quantum maps on the distinguished Orlicz spaces.

The paper is organized as follows: In the next section we set up notation and terminology. Moreover, for the convenience of the reader we repeat the relevant material from \cite{ML1}, thus making our exposition self-contained. In Section 3 we give a brief exposition of the theory of crossed products.
Dynamical maps on crossed products as well as on certain subsets of measurable operators will be considered in Section 4.
Section 5 is devoted to the study of quantum maps on the set of regular observables. Moreover, an exposition of quantum maps defined via embeddings will be provided. Section 6 establishes the relations between Dirichlet forms on non-commutative spaces and quantum maps on the studied Orlicz spaces. The last section contains some conclusions and remarks.

\section{Notation, terminology, and previous results}
We follow notation used in \cite{ML1} and \cite{ML2}. Let
$(X, \mu)$ be a measure space. We denote $L^1(X, \mu)= \{ f : \int_X |f|d\mu < \infty\}$, while $L^{\infty}(X, \mu)$ stands for the essentially bounded, measurable functions on $X$.
Their non-commutative analogues are: $\mathfrak{F}_T(\mathcal{H})$ - the trace class operators on a Hilbert space $\mathcal{H}$, and 
$B(\mathcal{H})$ -- all linear bounded operators on $\mathcal{H}$.
We remind the reader that
$L^p(X, \Sigma, m)$ spaces ($1\leq p < \infty$), $(X, \Sigma, m)$ a measure space,  may be regarded as spaces of measurable functions conditioned by the functions $t \mapsto |t|^p$ ($1\leq p < \infty$).
 The more general category of Orlicz spaces is defined as spaces of measurable functions conditioned by a more general class of convex functions; the so-called Young's functions.
 A function $\Phi: [0, \infty) \to [0, \infty]$ is a Young's function if $\Phi$ is convex, $\lim_{u \to 0+} \Phi(u) = \Phi(0) = 0$, $\lim_{u\to \infty}\Phi(u)=\infty$, and $\Phi$ is non-constant on $(0,\infty)$. It is worth pointing out that such functions have  a nice integral representation, for details see \cite{ML1} and the references given there.
 
 Let $L^0$ be the space of measurable 
functions on some $\sigma$-finite measure space $(X, \Sigma, \mu)$. We will always assume  that the considered measures
are $\sigma$-finite.
\begin{definition}
The Orlicz space 
$L^{\Psi}$ (being a Banach space) associated with  $\Psi$ is defined to be the set 
\begin{equation}\label{3}
L^{\Psi} \equiv L^{\Psi}(X, \Sigma, \mu) = \{f \in 
L^0 : \Psi(\lambda |f|) \in L^1 \quad \mbox{for some} \quad \lambda = \lambda(f) > 0\}.
\end{equation}
\end{definition}
The basic Orlicz spaces used in this paper are $L\log(L+1)$, and $L^{\cosh - 1}$ defined by Young's functions: $x \mapsto x \log(x+1)$,
and $x \mapsto \cosh(x) - 1$ respectively.
  
For the finite measure case, the  spaces $L\log(L+1)$ and $L^{\cosh -1}$ can be identified with Zygmund spaces. They are defined as follows (cf \cite{BS}):
\begin{itemize}
\item $L\log L$ is defined by the following Young's function
$$s\log^+ s = \int_0^s \phi(u) du$$
where $\phi(u) =0 $ for $0\leq u\leq1$ and $\phi(u) = 1 + \log  u$ for $1< \infty $, where $\log^+x = \max (\log x, 0)$
\item $L_{\exp}$ is defined by the Young's function
$$ \Psi(s) = \int_0^s \psi(u)du, $$
where $\psi(0) = 0$ , $\psi(u) = 1$ for $0<u<1$, and $\psi(u)$ is equal to $e^{u -1}$ for $1 < u < \infty$.
Thus
$\Psi(s) = s$ for $0\leq s \leq 1$ and  $\Psi(s) = e^{s - 1}$ for $1< s < \infty$.
\end{itemize}

There is a natural question: \textit{what can be said about uniqueness of the correspondence: Young's function $\Psi \mapsto L^{\Psi}$-Orlicz space}. To answer this question one needs the concept of equivalent Young's functions. To define it we will write $F_1 \succ F_2$ if and only if $F_1(bx) \geq F_2(x)$ for $x\geq 0$ and some $b>0$, and we say that the functions $F_1$ and $F_2$ are equivalent, $F_1 \approx
F_2$, if $F_1\prec F_2$ and $F_1\succ F_2$. One has (see \cite{RR})

\begin{theorem}
\label{2.6}
Let $\Phi_i$, $i =1,2$ be a pair of equivalent Young's function. Then $L^{\Phi_1} = L^{\Phi_2}$. 
\end{theorem}
Consequently, on condition that equivalence is preserved, \textit{one can ``manipulate'' Young's function's!}

Our main results concerning classical statistical physics, stated and proved in \cite{ML1} (see also \cite{ML2}) are:
\begin{theorem}
\label{klasyczne}
 The dual pair $\langle L^{\cosh -1}, L\log(L+1)\rangle$  provides the basic mathematical ingredient for a description of a general, classical regular system while,
 for the finite measure case, the above pair of Orlicz spaces is an equivalent renorming of the pair of Zygmund spaces $( L_{\exp}, L\log L)$.
\end{theorem}

Turning to the quantum case, as a first step, one should define the quantum counterpart of measurable functions $L^0$. In the quantum world there is no known space that is a direct analogue of $L^0$. But one is able to define a quantum analogue of the space of all measurable functions which are bounded, except on a set of finite measure. This space turns out to be more than adequate for our purposes. So to this end let $\mathfrak{M} \subset B(\mathcal{H})$ be a semifinite von Neumann algebras 
 equipped with an fns (faithful normal semifinite) trace $\tau$. The space of all $\tau$-measurable operators is defined as follows.
 Let $a$ be a densely defined closed operator on $\mathcal{H}$ with domain ${\mathcal D}(a)$ and let
$a = u|a|$ be its polar decomposition.
One says that $a$ is affiliated with $\mathfrak{M}$ (denoted $a \eta \mathfrak{M}$) if $u$ and all the spectral projections of $|a|$ belong to $\mathfrak{M}$.
Then $a$ is $\tau$-measurable if $a \eta \mathfrak{M}$, and for each $\delta >0$, there exists a projection $e \in \mathfrak{M}$ such that $e\mathcal{H} \subset {\mathcal D}(a)$ and $\tau(1 - e)\leq \delta$.
We denote by $\widetilde{\mathfrak{M}}$ the set of all $\tau$-measurable operators.
The algebra $\widetilde{\mathfrak{M}}$ (equipped with the topology of convergence in measure) 
is a substitute for 
$L^0$ in the quantum world (for details see \cite{nelson}, \cite{terp}, and \cite{se}).

Following the Dodds, Dodds, de Pagter approach \cite{DDdP} we need 
the concept of generalized singular values. Namely,
given an element $f \in \widetilde{\mathfrak{M}}$ and $t \in [0, \infty)$, the generalized singular 
value $\mu_t(f)$ is defined by $\mu_t(f) = \inf\{s \geq 0 : \tau(\jed - e_s(|f|)) \leq t\}$ 
where $e_s(|f|)$ $s \in \mathbb{R}$ is the spectral resolution of $|f|$. The function $t \to 
\mu_t(f)$ will generally be denoted by $\mu(f)$. For details on the generalized singular values 
see \cite{FK}.  Here, we note only that this directly extends classical notions where for any $f \in L^0{}$, 
the function $(0, \infty) \to [0, \infty] : t \to \mu_t(f)$ is known as the decreasing 
rearrangement of $f$. 

The key ingredient of the Dodds, Dodds, de Pagter approach is the concept of a Banach Function Space. To define this concept, let $L^0(0, \infty)$
stand for
measurable functions on $(0, \infty)$ and $L^0_+$ denote $\{ f \in L^0(0,  \infty); f \geq 0\}$. A function norm

$\rho$ on $L^0(0, \infty)$ is defined to be a mapping $\rho : L^0_+ \to [0, \infty]$ satisfying
\begin{itemize}
\item $\rho(f) = 0$ iff $f = 0$ a.e.  
\item $\rho(\lambda f) = \lambda\rho(f)$ for all $f \in L^0_+, \lambda > 0$.
\item $\rho(f + g) \leq \rho(f) + \rho(g)$ for all $f, g \in L^0_+$.
\item $f \leq g$ implies $\rho(f) \leq \rho(g)$ for all $f, g \in L^0_+$.
\end{itemize}
Such a $\rho$ may be extended to all of $L^0$ by setting $\rho(f) = \rho(|f|)$, in which case 
we may then define $L^{\rho}(0, \infty) = \{f \in L^0(0, \infty) : \rho(f) < \infty\}$. If 
now $L^{\rho}(0, \infty)$ turns out to be a Banach space when equipped with the norm 
$\rho(\cdot)$, we refer to it as a Banach Function space. If $\rho(f) \leq \lim\inf_n\rho(f_n)$ 
whenever $(f_n) \subset L^0$ converges almost everywhere to $f \in L^0$, we say that $\rho$ 
has the Fatou Property. If less generally this implication only holds for $(f_n) \cup \{f\} 
\subset L^{\rho}$, we say that $\rho$ is lower semi-continuous. If further the situation $f 
\in L^\rho$, $g \in L^0$ and $\mu_t(f) = \mu_t(g)$ for all $t > 0$, forces $g \in L^\rho$ and 
$\rho(g) = \rho(f)$, we call $L^{\rho}$ rearrangement invariant (or symmetric).

By employing  generalized singular values and Banach Function Spaces, 
 Dodds, Dodds and de Pagter \cite{DDdP} formally defined the noncommutative space 
$L^\rho(\widetilde{\mathfrak{M}}) \equiv L^\rho(\mathfrak{M}, \tau) \equiv L^\rho(\mathfrak{M})$ to be  $$L^\rho({\mathfrak{M}}) = \{f \in \widetilde{\mathfrak{M}} : \mu(f) \in 
L^{\rho}(0, \infty)\}$$ and showed that if $\rho$ is lower semicontinuous and $L^{\rho}(0, 
\infty)$ rearrangement-invariant, $L^\rho({\mathfrak{M}})$ is a Banach space when equipped 
with the norm $\|f\|_\rho = \rho(\mu(f))$. 

Having quantized Orlicz spaces we showed that in the quantum context, a quantized version of Theorem \ref{klasyczne} is valid (see \cite{ML1} and \cite{ML2} for details).

\begin{theorem}
\label{kwantowe}
 The dual pair of quantum Orlicz spaces $(L^{\cosh -1}, L\log(L+1))$  provides the basic mathematical ingredient for a description of a general, quantum regular system while,
 for the finite measure case, the above pair of Orlicz spaces is an equivalent renorming of the pair of quantum Zygmund spaces $( L_{\exp}, L\log L)$.
\end{theorem}

\section{Crossed products}

In contexts where we are dealing with a von Neumann algebra which does not have a trace, we do not have access to the elegant theory of Dodds, Dodds, and de Pagter. In such cases we will follow the philosophy of Haagerup and Terp, which makes essential use of the notion of crossed products of von Neumann algebras to posit quantum $L^p$-spaces. For the sake of the reader we briefly review some of the essential facts regarding continuous crossed products, before going on to analyse the behaviour of quantum dynamical maps with respect to these crossed products.

Let $\mathfrak{M}$ be a $\sigma$-finite von Neumann algebra acting on a separable Hilbert space $\cH$. 
A fixed faithful normal state $\omega$ on $\mathfrak{M}$ will be defined by $\omega(x) = (\Omega, x \ \Omega)$  where  $\Omega \in \cH$ is a cyclic and separating vector.
Thus, we will be concerned with the standard form
\begin{equation}
(\mathfrak{M}, \cH, \cP, J, \Omega),
\end{equation}
where $\cP$ ($J$) is  the natural cone (modular conjugation respectively). The modular automorphism $\sigma_t$
of $\mathfrak M$ is given by $\sigma_t(\cdot) = \Delta^{it}x \Delta^{-it}$, where $t \in \Rn$, $x \in \mathfrak M$ and $\Delta$ is the modular operator.

The crossed product algebra $\mathfrak{M} \rtimes_{\sigma} \Rn$ is a von Neumann algebra, acting on $L^2(\Rn, \cH)$, and generated by operators
$\pi(x)$ and $\lambda(t)$, defined by (cf \cite{vD}, \cite{KR})
\begin{equation}
(\pi(x)\xi)(s) = \sigma_{-s}(x) \xi(s),
\end{equation}
\begin{equation}
(\lambda(t) \xi)(s) = \xi(s - t),
\end{equation}
where $x \in \mathfrak M$, $t \in \Rn$, and $\xi \in C_c(\Rn, \cH)$ - the space of continuous functions on $\Rn$ with values in $\cH$ and compact supports.

It follows from the definition of the crossed product that (cf \cite{vD})
\begin{enumerate}
\item $\lambda(t) \pi(x) \lambda^*(t) = \pi(\sigma_t(x))$ for $x \in \mathfrak{M}$ and $t \in \Rn$.
\item $\pi(x)\lambda(t) \pi(y) \lambda(s) = \pi(x\sigma_t(y)) \lambda(ts)$,
\item $(\pi(x) \lambda(t))^* = \pi(\sigma_{-t}(x^*)) \lambda(-t)$,
\item $\mathfrak{M} \rtimes_{\sigma} \Rn$ is the closure of the $^*$-algebra of linear combinations of products $\lambda(s)\pi(x)$ with $x \in \mathfrak M$ and $s \in \Rn$.
\end{enumerate}
To clarify the definition of crossed products given above, we wish to make the following remarks:

\begin{remark}
\begin{enumerate}
\item $L^2(\Rn, \cH)$ can be canonically identified with $\cH \otimes L^2(\Rn)$ by
\begin{equation}
(U(\xi_0 \otimes f))(s) = f(s) \xi_0,
\end{equation}
for any $\xi_0 \in \cH$ and $f \in C_c(\Rn)$ ($C_c(\Rn)$ - the space of continuous complex valued functions on $\Rn$ with compact supports), see Proposition 2.2 in \cite{vD}.
\item
Let $\lambda_t$ denote the left translation by $- t$ in $L^2(\Rn)$. Then , see  Proposition 2.8 in \cite{vD}, 
\begin{equation}
\lambda(t)= \jed \otimes \lambda_t.
\end{equation}
\item $\mathfrak{M} \rtimes_{\sigma} \Rn$ is spatially isomorphic to the von Neumann algebra on $\cH \otimes L^2(\Rn)$ generated by the operators 
$$ \{ x \otimes \jed, \Delta^{is} \otimes \lambda_s; x \in \mathfrak{M}, s \in \Rn \},$$
see Proposition 2.12 in \cite{vD}.
\item The dual action $\widehat{\sigma}$ of the dual group $\widehat{\Rn} \equiv \Rn$ on $\mathfrak{M} \rtimes_{\sigma} \Rn$ is defined as
\begin{equation}
\label{1.6}
\widehat{\sigma_t}(a) = (\jed \otimes v_t)a (\jed \otimes v^*_t),
\end{equation}
where $a \in \mathfrak{M} \rtimes_{\sigma} \Rn$, and $(v_tf)(s) = e^{-its} f(s)$ for any $f \in C_c(\Rn)$.
\end{enumerate}
\end{remark}

It is essential to note that $\mathfrak{M} \rtimes_{\sigma} \Rn$ \textit{is a semifinite  von Neumann algebra}, and the canonical faithful semi-finite trace is defined on $(\mathfrak{M} \rtimes_{\sigma} \Rn)^+$ is
\begin{equation}
\label{1.7}
\tau = \sup_K \tau_K,
\end{equation}
where $\tau_K(a) = (\xi_K, a \xi_K)$, $a \in \mathfrak{M} \rtimes_{\sigma} \Rn$. Here  $\xi_K = \Omega \otimes \cF^* f_K$, where $f_K(s) = \chi_K(s) \exp{\frac{s}{2}}$, and $\cF$ stands for the Fourier transform on $L^2(\Rn)$. $K$ denotes a compact subset in $\Rn$.

Furthermore, see Lemma 3.3 in \cite{vD}, if $f \in C_c(\Rn)$ and it has support in the compact set $K$, then 
\begin{equation}
\label{1.8}
\tau_K(\pi(x) \lambda(s) \lambda(f)) = 2 \pi \widehat{f}(i + s) \omega(x),
\end{equation}
where $\lambda(f) = \jed \otimes \lambda_f$, $\lambda_f = \cF^* m_f \cF$, and $m_f$ is the multiplication operator by $f$ in $L^2(\Rn)$.

We wish to end these preliminaries with:
\begin{remark}
Although one has (cf Lemma 3.1 in \cite{vD})
\begin{equation}
\mathfrak{M} \rtimes_{\sigma} \Rn \subseteq \mathfrak{M} \otimes \cB(L^2(\Rn)),
\end{equation}
one can not expect, in general, that $\mathfrak{M} \rtimes_{\sigma} \Rn$ would be of the form $\mathfrak{M} \otimes \mathfrak{N}$ for a von Neumann algebra $\mathfrak{N}$. The simplest argument supporting this claim is that then $\mathfrak{M} \rtimes_{\sigma} \Rn$ would be of type III for any $\mathfrak{M}$ of type III (see Table 11.2 as well as Proposition 11.2.26 in \cite{KR}) which would contradict (\ref{1.7}) and the semifinitness of $\mathfrak{M} \rtimes_{\sigma} \Rn$.
\end{remark}

\section{Dynamical maps on measurable operators}
\subsection{Detailed Balance Condition}
\label{DBCo}
Let $T: \mathfrak{M} \to \mathfrak{M}$ be a positive, normal, unital map. Such maps will be called \textit{Markov maps}. The class of Markov maps seems to be too general to describe the most interesting genuine dynamics.
Hence, to select more regular maps we define:
\begin{definition}
\label{DBC}
A Markov map satisfies the Detailed Balance Condition (for brevity DBC) with respect to a state $\omega$ on $\mathfrak M$ if the following conditions are satisfied (see  \cite{Maj2}, \cite{Maj1})
\begin{equation}
\omega(x^*T(y)) = \omega(\Theta(y^*) T \Theta(x))
\end{equation}
for any $x,y \in \mathfrak M$, where $\Theta$ is a reversing operation, i.e. an antilinear Jordan morphism on $\mathfrak M$ such that $\Theta^2 = \rm{identity \ map}$, and $\omega(\Theta(xy)) = \omega(\Theta(x) \Theta(y))$.
\end{definition}

DBC implies that (see \cite{Maj1}):
\begin{equation}
\label{1.10}
\omega(T(x)) = \omega(x), \quad x \in \mathfrak M.
\end{equation}
and that 
\begin{equation}
\label{for1.12}
\widehat{T} x \Omega = T(x) \Omega
\end{equation}
defines a bounded operator on $\cH$ which commutes with the modular operator $\Delta$.
Moreover, it is an easy observation to make that 
\begin{equation}
T \circ \sigma_t = \sigma_t \circ T, \quad \rm{ for \ any} \ t \in \Rn.
\end{equation}
To see this note that for any $x \in \mathfrak{M}$ and any $y^{\prime} \in \mathfrak{M}^{\prime}$ ($\mathfrak{M}^{\prime}$ stands for the commutator of $\mathfrak{M}$) one has:
$$T(\sigma_t(x))y^{\prime} \Omega = y{\prime} T(\sigma_t(x))\Omega 
= y^{\prime} \widehat{T} \sigma_t(x) \Omega = y^{\prime} \widehat{T} \Delta^{it} x \Omega = y^{\prime} \sigma_t(T(x)) \Omega = \sigma_t(T(x)) y^{\prime} \Omega$$ 
which proves the claim.

Before proceeding further let us pause to make some important remarks on the DBC.

\begin{remark}
\begin{enumerate}
\item There are various versions of DBC. For example, one can use the more general form of DBC which was given in \cite{MS}. However, the form given here has a more ``transparent'' physical interpretation.
In particular, to the best of our knowledge, only the form of DBC given in Definition \ref{DBC} leads to a one-to-one correspondence between dynamical semigroups on the set of observables and semigroups on the Hilbert space of (state) vectors respectively (see \cite{Maj1}). 
\item Frequently, DBC is related to KMS symmetry. However it is important to note that only DBC forces the map $T$ to commute with the authormorphism group, and this property will be essential in our analysis.
\item \textbf{But}, in general, tensor product structure is not respected by DBC. Namely, if a (positive) map $T: \mathfrak{M} \to \mathfrak{M}$ satisfies DBC then  $T \otimes id: \mathfrak{M} \otimes \mathfrak{N} \to \mathfrak{M} \otimes \mathfrak{N}$, where $\mathfrak{N}$ is a $^*$-algebra, does not need to be a positive map. Therefore,
one can not expect that an extension of a positive map  $T$ on the tensor product structure will satisfy DBC. Consequently, to get well defined dynamical maps on the crossed products, a further selection of positive maps should be done. To this end, \textit{ complete positivity will be assumed additionally}. 
\item For a recent account on DBC we refer the reader to \cite{FU}.
\end{enumerate}
\end{remark}

\subsection{Extension of dynamical maps to crossed products and their corresponding algebra of $\tau$-measurable operators}

\vskip 1cm

In our study of quantum maps, we need to canonically extend a dynamical map $T$ defined on a von Neumann algebra $\mathfrak M$ to a corresponding map which is defined on a certain noncommutative Orlicz space. As a first step we have to extend $T$ to the corresponding crossed product. Here, we will follow the definition given in \cite{HJX}.

\begin{definition}
\label{1.5}
Let $T:\mathfrak{M} \to \mathfrak{M}$ be a positive, unital map satisfying DBC with respect to a faithful, normal state $\omega(\cdot) = (\Omega, \cdot \ \Omega)$.
Define
\begin{equation}
\tT(\lambda(t) \pi(x)) = \lambda(t) \pi(T(x)),
\end{equation}
for $t \in \Rn$, and $x \in \mathfrak{M}$.
\end{definition}

\begin{remark}
\label{prob1.6}
To have a well defined linear map $\tT$ on $\mathfrak{M} \rtimes_{\sigma} \Rn$ one wishes to have
\begin{equation}
\tT(\widetilde{x}) = \tT(\sum_i \lambda(s_i) \pi(x_i)) = \sum_i \lambda(s_i) \pi(T(x_i)),
\end{equation}
where $\widetilde{x} = \sum_i \lambda(s_i) \pi(x_i) \in \mathfrak{M} \rtimes_{\sigma} \Rn$.
But to guarantee the well definiteness of the linear map $\tT$, so to have $\tT(0) = 0$, one should be able to show that
\begin{equation}
\label{1.16}
\sum_i \lambda(s_i) \pi(x_i) = 0
\end{equation}
implies
\begin{equation}
\sum_i \lambda(s_i) \pi(T(x_i)) = 0
\end{equation}
for $s_i \in \Rn$ and $x_i \in \mathfrak{M}$.
To this end let us consider (\ref{1.16}) in detail. Namely, note that (\ref{1.16}) implies $\|\sum_i \lambda(s_i) \pi(x_i) \xi\| = 0$ for any $\xi \in L^2(\Rn, \cH)$. Taking $\xi(t)$ to be of the form $\xi(t) = f(t)\Omega $ with $f \in L^2(\Rn)$, one has
\begin{eqnarray}
0 &=& \int_{\Rn} \sum_{i,j}\Big(\lambda(s_i) \pi(x_i) f(t) \Omega, \lambda(s_j) \pi(x_j) f(t) \Omega\Big)dt\\
&=& \int_{\Rn} \sum_{i,j} \big(\sigma_{s_i}(x_i) \Omega, \sigma_{s_j}(x_j) \Omega\big) \overline{f(t - s_i)} f(t - s_j) dt\nonumber\\
&=& \int_{\Rn} \sum_{i,j} \omega(\sigma_{s_i}(x^*_i)\sigma_{s_j}(x_j)) \overline{f(t - s_i)} f(t - s_j) dt.\nonumber
\end{eqnarray}
Further note that on setting $a_{i,j} \equiv \int_{\Rn}\overline{f(t - s_i)} f(t - s_j) dt$,
one obtains the positive definite matrix $a_{i,j}$. As $f \in L^2(\Rn)$ is an arbitrary function one can then expect that $a_{i,j}$ is an arbitrary positive defined matrix. 
As the matrix $a$ is positive $a \equiv \{a_{i,j} \} \geq 0$, it should then be of the form $a = b^* b$, see Lemma 3.1 in Chapter IV \cite{Tak}. Therefore $a_{i,j} = \sum_k\overline{b_{k,i}} \ b_{k,j}.$
Consequently, the fact that $\|\sum_i \lambda(s_i) \pi(x_i) \xi\| = 0$, combined with our choice of the vector $\xi$ leads to
\begin{eqnarray}
0 &=& \sum_{i,j} \omega(\sigma_{s_i}(x^*_i)\sigma_{s_j}(x_j)) a_{i,j} = \sum_{i,j,k}  \omega(\sigma_{s_i}(x^*_i)\sigma_{s_j}( x_j)) \overline{b_{k,i}} \ b_{k,j}\\
& =& \sum_k \omega(\Big(\sum_i b_{k,i} \sigma_{s_i}(x_i)\Big)^* \Big(\sum_j b_{k,j} \sigma_{s_j}(x_j)\Big)).\nonumber
\end{eqnarray}
But, $\Omega$ is cyclic and separating, so $\omega$ is a faithful state. Thus, one gets 
\begin{equation}
\label{1.20a}
\sum_i b_{k,i} \sigma_{s_i}(x_i) = 0.
\end{equation}
To see how big is the family of matrices $\{b_{i,j} \}$ let us take a basis in $L^2(\Rn)$, for example consisting of $H_n(t)$-Hermite polynomials, and note that
\begin{equation}
a_{i,j} = \sum_n \int_{\Rn} \overline{f(t-s_i)} H_n(t) dt \int_{\Rn} \overline{H_n(t)} f(t-s_j) dt \equiv \sum_n \overline{b_{n,i}} b_{n,j}.
\end{equation}
Then (\ref{1.20a}) can be rewritten as
\begin{equation}
\sum_i \int_{\Rn} \overline{H_n(t)} f(t-s_i) dt \ \sigma_{s_i}(x_i) = 0
\end{equation}
for any $f \in L^2(\Rn)$ and any $n$. 

However, we note that $\{b_{n,i} \equiv \int_{\Rn} \overline{H_n(t)} f(t-s_i) dt \}$, for fixed $i$,  is an arbitrary element in $l_2$-space. Thus, we see at once that $\{\lambda(s_i) \pi(x_i) \}$ are linearly independent and the map $\tT$ is well defined.
\end{remark}
\vskip 1cm

To formulate and then to prove our first results concerning $\tT$, we need some preliminaries. Firstly note that if $x \in \mathfrak M$ then we can define an operator $\tx$ on $L^2(\Rn, \cH)$ by $(\tx \xi)(s)= x \xi(s)$ for $\xi \in C_c(\Rn, \cH)$ (see \cite{vD}). The important point to make here is that the form of $\xi = \xi_0 \otimes f$ with $\xi_0 \in \cH$ and $f \in C_c(\Rn)$ leads to
\begin{equation}
x\xi(s) = xf(s) \xi_0 = f(s) x \xi_0 = (x\xi_0 \otimes f)(s).
\end{equation}
Consequently, $\tx = x \otimes \jed$ and $x \otimes \jed$ maps $C_c(\Rn, \cH)$ into $C_c(\Rn, \cH)$. Furthermore, one can write
\begin{equation}
((x \otimes \jed) \xi)(s) = x \xi(s).
\end{equation}

Turning to measurable operators, we recall that $\cM \equiv \mathfrak{M} \rtimes_{\sigma} \Rn$ is a semifinite von Neumann algebra equipped with a canonical normal faithful semifinite trace $\tau$ (cf Section 3). Therefore, one can define the family of measurable operators $\widetilde{\cM}$, see \cite{terp}, \cite{se}, \cite{nelson}. We will need (cf \cite{terp}) the following two basic facts established in the theory of noncommutative $L^p$-spaces. Firstly $\widetilde{\cM}$ is a complete Hausdorff topological $^*$-algebra in which $\cM \equiv \mathfrak{M} \rtimes_{\sigma} \Rn$ is dense (see Theorem 28 in \cite{terp}).

Secondly, we remind the reader that (Haagerup's) $L^p$-space can be considered to be the result of a selection of measurable operators from $\widetilde{\cM}$ which are ``p-homogeneous'' with respect to the dual action $\widehat{\sigma}$ of $\Rn$. More precisely
\begin{equation}
L^p(\mathfrak{M}) = \{ a \in \widetilde{\cM}; \ \forall_{s \in \Rn} \quad \widehat{\sigma}_s a = e^{- \frac{s}{p}} a \}.
\end{equation}

Now, we are in a position to formulate and to prove our first result.

\begin{theorem}
\label{4.5a}
Let $T:\mathfrak{M} \to \mathfrak{M}$ be a completely positive, unital map satisfying DBC with respect to a faithful, normal state $\omega(\cdot) = (\Omega, \cdot \ \Omega)$. Then
\begin{enumerate}
\item 
$\tT$ is a bounded linear map on $\mathfrak{M} \rtimes_{\sigma} \Rn$.
\item $\tau \circ \tT = \tau$.
\end{enumerate}
\end{theorem}
\begin{proof}
We remind the reader that any completely positive map $T$ is also completely bounded and $\| T(\jed)\| = \|T\| = \|T\|_{cb}$, where $\| \cdot \|_{cb}$ stands for the cb-norm, see Proposition 3.6 in \cite{vP}. 
Moreover, all finite linear combinations of 
$\lambda(s)\pi(x)$, $s \in \Rn$ $x \in \mathfrak{M}$ form a $^*$-dense involutive subalgebra of $\mathfrak{M} \rtimes_{\sigma} \Rn$.
Thus, the statement that $\tT$ is
the well defined map on $\mathfrak{M} \rtimes_{\sigma} \Rn$ follows from Theorem 4.1 in \cite{HJX} and this proves the first claim.

To prove the second claim, firstly note (cf (\ref{1.8})) that for $f = \chi_K$ ($\chi_K$ stands for the indicator function of a compact subset $K \subset \Rn$) one has
\begin{eqnarray}
\label{1.18}
\tau_K(\lambda(s)\pi(x) \lambda(\chi_K)) &=& \tau_K(\pi(\sigma_s(x)) \lambda(s) \lambda(\chi_K))\\
&=& \omega(\sigma_s(x)) \int e^{-its} e^t \chi_K(t) dt\nonumber\\ 
&=& \omega(x) \int_K e^{-its} e^t \chi_K(t) dt\nonumber\\
&=& \tau_K(\lambda(s) \pi(x)),\nonumber
\end{eqnarray}
where we have used the invariance of $\omega$ with respect to the modular automorphism and the formula (\ref{1.8}).
The last equality follows from  
$$\lambda(\chi_K)\xi_K = \jed \otimes \cF^* m_{\chi_K}\cF  \cdot \   \Omega \otimes \cF^* f_K = \xi_K
$$
(cf Section 3). Secondly
\begin{equation}
\label{1.27}
\tau_K \circ \tT(\lambda(s) \pi(x)) = \tau_K(\lambda(s) \pi(T(x))) = \tau_K (\lambda(s) \pi(x)),
\end{equation}
where we have used Definition \ref{1.5}, (\ref{1.18}), and the invariance of $\omega$ with respect to $T$.
Thirdly, we note that
\begin{equation}
(\lambda(s) \pi(x))^*\lambda(s) \pi(x) = \pi(x^*x),
\end{equation}
and hence as $\lambda(t) \pi(x) \lambda(t)^* = \pi (\sigma_t(x))$:
\begin{eqnarray}
\Big(\sum_i\lambda(s_i) \pi(x_i)\Big)^*\Big(\sum_j \lambda(s_j) \pi(x_j)\Big)
&=& \sum_{i,j} \pi(x_i)^* \lambda(s_j - s_i) \pi(x_j)\\
&=& \sum_{ij} \lambda(s_j-s_i) \pi(\sigma_{s_i - s_j}(x_i)^*x_j),\nonumber
\end{eqnarray}
for any $s, t \in \Rn$ and $x, y \in \mathfrak{M}$. 
Define $\tT$ as in Remark \ref{prob1.6}, i.e
\begin{equation}
\tT(\widetilde{x}) = \tT(\sum_i \lambda(s_i) \pi(x_i)) = \sum_i \lambda(s_i) \pi(T(x_i)),
\end{equation}
where $\widetilde{x} = \sum_i \lambda(s_i) \pi(x_i) \in \mathfrak{M} \rtimes_{\sigma} \Rn$.
As
\begin{equation}
\widetilde{x}^* \widetilde{x} = (\sum_i \lambda(s_i) \pi(x_i))^*(\sum_j \lambda(s_j) \pi(x_j))
\end{equation}
$$ = \sum_{ij} \lambda(s_j-s_i) \pi(\sigma_{s_i - s_j}(x_i)^*x_j),
$$
one has
\begin{eqnarray}
\tau_K \circ \tT(\widetilde{x}^* \widetilde{x}) &=& \tau_K(\sum_{i,j} \lambda(s_i -s_j) \pi \circ T (\sigma_{s_i - s_j}(x_i)^* x_j))\\
&=& \sum_{i,j} \tau_K(\lambda(s_i -s_j) \pi(\sigma_{s_i - s_j}(x_i)^* x_j) = \tau_K(\widetilde{x}^* \widetilde{x}),\nonumber
\end{eqnarray}
where we have used (\ref{1.27}). But since $\tau_K$ increases to the trace $\tau$ over  $\mathfrak{M} \rtimes_{\sigma} \Rn$ as $K$ increases, the trace $\tau$ is also invariant with respect to $\tT$.
\end{proof}

\section{Quantum maps on the set of regular observables}\label{1.4s}

The basic aim of this section is a description of quantum maps on the Orlicz spaces $L^{\cosh - 1}$
and $L_{\exp}$. We remind the reader that these spaces
are intended to describe regular observables (see \cite{LM}, \cite{ML1}, and \cite{ML2}). Here, we wish to show that the important class of quantum maps originally defined on the von Neumann algebra generated by bounded observables, give well defined time evolution of regular observables.

As it was mentioned in the Introduction, the main difficulty in carrying out the description of time evolution of regular observables is that one needs to adapt duality of time evolution of quantum systems (Heisenberg and Schr\"odinger pictures) to the extension of the interpolation scheme for noncommutative spaces. To clarify the picture, we will for the reader's convenience, make some preliminary observations:
\begin{observation}
\label{1.11a}
\begin{enumerate}
\item Classically, interpolation spaces are intermediate spaces between the sum and the intersection of a pair of Banach spaces. Hence in the category of $L^p$ spaces, a natural candidate for such a pair of Banach spaces should be
the spaces $L^1$ and $L^{\infty}$. Moreover, the space $L^1 + L^{\infty}$ is large enough to accommodate all interesting Orlicz spaces, and $L^1\cap L^\infty$ small enough to live in the intersection of these spaces. However, the pair $\langle L^1 \cap L^{\infty},L^1 + L^{\infty}\rangle$ is not a Calder\'on couple, see \cite{MO}. In other words, this pair does not have all the nice properties of such a couple, and hence special care is necessary in choosing the correct pair of spaces to start with.
\item For the description of the pair $\langle L^1 \cap L^{\infty},L^1 + L^{\infty}\rangle$ in the general noncommutative setting see \cite{L}.
\item In contrast to the pair $\langle L^1 \cap L^{\infty},L^1 + L^{\infty}\rangle$, the interpolation spaces for the couple $\langle L^1, L^{\infty}\rangle$ coincide with the rearrangement invariant spaces on $\Rn$ and $\langle L^1, L^{\infty}\rangle$ is a Calder\'on couple (for all details see Chapter 26, \textit{Interpolation of Banach spaces} by N. Kalton, S. Montgomery-Smith in \cite{JL}).
\item In particular, in the general noncommutative context, see \cite{L} Theorem 3.13, on the spaces $L^{\cosh - 1}$ 
and $L_{\exp}$ the topology of convergence in measure is normable while it seems that this is not the case for the corresponding dual spaces, i.e. $L\log(L+1)$ and $L\log L$.
\item The important point to note here is the role of the Hilbert space $L^2(\mathfrak{M})$. Namely, a large class of dynamical maps is defined in terms of Dirichlet forms on such Hilbert space, see \cite{GIS}, and \cite{Cip} for more details and an comprehensive bibliography. Thus, this case will be treated in a separate section.
\item Furthermore, the case of the Hilbert space $L^2(\mathfrak{M})$ gains in interest if we realize that the quantization of Markov-Feller processes can be done within the framework of noncommutative Hilbert spaces, see \cite{MZ1}, \cite{MZ2}, \cite{MZ3}, \cite{MZ4}. 
\item Hence, to get a simpler description of quantum maps arising both from Dirichlet forms as well as from bounded maps on $L^2(\mathfrak{M})$, we combine some ideas given in \cite{LS} with the framework outlined in the previous sections.
\end{enumerate}
\end{observation}

Before proceeding with the study of the interpolation in the non-commutative setting, let us pause to clarify the description of the relevant quantum spaces.

The set of observables of a given quantum system leads to a $\sigma$-finite von Neumann algebra $\mathfrak{M}$. To have non-commutative measurability, the large semifinite von Neumann algebra $\cM = \mathfrak{M}\rtimes_{\sigma}\mathbb{R}$ should be employed, see Sections 2 and 3. We emphasize that the original set of observables, described by $\mathfrak{M}$, can be identified in the larger algebra $\cM$, see Remark 6.1 in \cite{ML1}.

Then, using non-commutative integration theory, there are essentially two ways of producing the kind of quantum spaces we are interested in. The first one is based on the DDdP approach, cf Section 2. However this theory is only directly available for the case where $\mathfrak{M}$ is semifinite. In the context of semifinite algebras, this approach was used in \cite{LM} and \cite{ML1} to describe quantum Orlicz spaces $L^{\Psi}(\mathfrak{M})$ of regular observables. Since $\cM$ is semifinite, we may of course in the present setting apply this theory to the action of the map $\widetilde{T}$ on $\cM$. However this is not an ideal solution, since we want quantum spaces that are more intrinsically related to $\mathfrak{M}$, not $\cM$.

The second type is based on Haagerup's strategy. Here a different way of selecting measurable operators in $\widetilde{\cM}$ is used, see \cite{terp}, \cite{L}. 
To describe this way we need some preliminaries. Let $\mathfrak{M}$ be a von Neumann algebra with $fns$ weight $\nu$. Further, let $h = \frac{d\tilde{\nu}}{d\tau}$ where $\tilde{\nu}$ is the dual weight of $\nu$ on the crossed product $\cM$ and $\tau \equiv \tau_{\cM}$ is its canonical trace. We will write $\mathfrak{n}_{\nu}$ for $\{a \in \mathfrak{M};  \nu(a^*a) < \infty\}$. Given an Orlicz space $L^{\Psi}(\Rn)$, the fundamental function induced by its Orlicz norm $\|\cdot \|_{\Psi}$ is defined by $\tilde{\varphi}_{\Psi}(t) = \| \chi_E \|_{\Psi}$ where $E$ is a measurable subset of $(\Rn, \lambda)$ for which $\lambda(E) = t$. Finally, a complementary Orlicz function $\Psi^*$ is defined by $\Psi^*(u) = \sup_{v>0}\{ uv - \Psi(v)\}$.
Using a different type of selecting measurable operators we arrived at (see \cite{L})
\begin{definition}\
We define the Orlicz space $L^{\Psi}(\mathfrak{M})$ to be
$$ L^{\Psi}(\mathfrak{M}) = \{ a \in \tilde{\cM}: [e\tilde{\varphi}_{\Psi^*}(h)^{\frac{1}{2}}]a[\tilde{\varphi}_{\Psi^*}(h)^{\frac{1}{2}} f] \in L^1(\mathfrak{M}) \rm{\ for \ all \ projections} \ e,f \in \mathfrak{n}_{\nu} \}.$$
\end{definition}

In the case where $\mathfrak{M}$ is semifinite and the canonical weight a trace, it is common to denote the associated spaces by either of $L^{\Psi}(\mathfrak{M}, \tau)$, or $L^{\Psi}(\widetilde{\mathfrak{M}})$. If Haagerup's strategy is used it is more common to write $L^{\Psi}(\mathfrak{M})$ as we did.

\begin{remark}
\label{5.3}
It is clear from the theory of quantum Banach function spaces developed by Dodds, Dodds, de Pagter, Sukochev, and many others, that although quantum, these spaces have deep connections with their classical counterparts, with their structure often exhibiting a close parallel to the classical versions living on the measure space $((0,\infty),\lambda)$. See for example \cite{DDdP2}. Hence here the process of quantisation is not only somehow clearer, but there is also much more structure one can work with. However as we noted earlier, this process is only directly available if $\mathfrak{M}$ is semifinite, and the canonical weight a trace. In the case of type III algebras the process of quantisation of these spaces is less clear, and some powerful machinery is needed to achieve a similar outcome. In the present setting we have started with an a priori given quantum map. Hence all that is needed, is clarity on quantising Orlicz spaces and an understanding of how to extend the action of the given map to such spaces. We proceed to review, describe, and where necessary, develop the requisite machinery.
\end{remark}

Now we turn to the second type of quantum spaces. It is worth pointing out that a major reason for the examination of both types of approaches, is the fact that, in the general case, there is as yet no quantized version of the real method of interpolation that we can fall back on. Hence here we need to make do with the complex method. However it is useful to know that in the semifinite case, we can call on real interpolation. If we are primarily interested in demonstrating the existence of a quantum map on the $L^p(\mathfrak{M})$-spaces ($1\leq p\leq\infty$), then Theorem 5.1 of \cite{HJX} will suffice.

However we are interested in demonstrating the existence of quantum maps on spaces quite different from $L^p$-spaces (here and subsequently in this section $L^p$ stands for Haagerup's $L^p$ space). In proving that we do have such a map on the space of regular observables, the primary difficulty we need to overcome, is that the current versions of the complex method only really work for $L^p$ spaces. So some ingenuity is needed if we are to be successful. The assumption that seems to help to bridge this gap, is the requirement that $T$ also be completely positive. We pause to point out that all of the theory developed in this section holds true for general von Neumann algebras, not just $\sigma$-finite ones. So unless otherwise stated, we will for the remainder of this section assume that $\mathfrak{M}$ is a possibly non-$\sigma$-finite algebra equipped with a faithful normal semifinite weight $\nu$. We will write $\sigma^{\nu}_t$ for the modular group associated with the weight $\nu$.

\begin{theorem}\label{Thm1} Let $\mathfrak{M}$ be as before, and assume that $T : \mathfrak{M} \to \mathfrak{M}$ is a
a completely bounded normal map such that
$$ T \circ \sigma^\nu_t=\sigma^\nu_t\circ T, \quad t \in \mathbb{R}.$$
Then $T$ admits a unique bounded normal extension $\widetilde{T}$ on $\mathfrak{M}\rtimes_{\sigma^\nu}\mathbb{R}$ such that 
$\|T\|=\|\widetilde{T}\|$ and
$$\widetilde{T}(\lambda(s)\pi_\alpha(x))= \lambda(s)\pi_\alpha(T(x)), \quad x\in m, s\in \mathbb{R}.$$
Moreover, $\widetilde{T}$ satisfies the following properties:
\begin{enumerate}
\item Let $B$ be the von Neumann subalgebra on $L^2(\mathbb{R},H)$ generated by all $\lambda(s)$, $s\in \mathbb{R}$. Then
$$\widetilde{T}(a\pi_\alpha(x)b)=a\pi_\alpha(T(x))b \quad \mbox{for all }a,b \in B.$$
\item $\widehat{T} \circ \sigma^{\widehat{\nu}}_t=\sigma^{\widehat{\nu}}_t\circ \widetilde{T}$ for all $t\in \mathbb{R}$
where $\widehat{\nu}$ is the dual weight of $\nu$.
\item If $T$ is positive, then so is $\widetilde{T}$.
\item Assume in addition that $\nu\circ T \leq \nu$. Then $\widehat{\nu}\circ \widetilde{T} \leq \widehat{\nu}$.
\end{enumerate}
\end{theorem}

\begin{proof} Modify the proof of \cite[Theorem 4.1]{HJX}.
\end{proof}

\begin{corollary}\label{Ttrace}
Let $T$ and $\widetilde{T}$ be as before. If each of (1)-(4) holds, then $\tau\circ \widetilde{T}\leq \tau$ where $\tau$ is the canonical trace on $\cM=\mathfrak{M}\rtimes_{\sigma^{\nu}}{\mathbb R}$.
\end{corollary}

\begin{proof}  Let $h=\frac{d\widehat{\nu}}{d\tau}$. 
By equation (1.1) of \cite{HJX} and page 2130 of \cite{HJX}, the action of $\sigma^\nu_t$ is induced by $a\to h^{it}ah^{-it}$. In the language of \cite{PT}, $\widehat{\nu}$ is then of the form $\tau(h\cdot) = \widehat{\nu}(\cdot)$ (see \cite[Theorem 5.12]{PT}). So by \cite[Proposition 4.3]{PT}, we have that $\tau=\widehat{\nu}(h^{-1}\cdot)$. From the proof of Theorem 7.4 of \cite{PT} considered alongside the discussion following Proposition 4.1 of \cite{PT}, it is clear that this means that for any $a\in (\mathfrak{M}\rtimes_{\sigma^{\nu}}{\mathbb R})_+$, have that 
$\tau(a)=\lim_{\epsilon\searrow 0}\widehat{\nu}([h^{-1/2}(\I+\epsilon h^{-1})^{-1/2}]a[h^{-1/2}(\I+\epsilon h^{-1})^{-1/2}])$. By the Borel functional calculus for affiliated operators, we have that $h^{-1}(\I+\epsilon h^{-1})^{-1}=(h+\epsilon\I)^{-1}$ for each $\epsilon>0$. So this formula becomes $\tau(a)=\lim_{\epsilon\searrow 0}\widehat{\nu}((h+\epsilon\I)^{-1/2}a(h+\epsilon\I)^{-1/2})$.
We also have that the von Neumann algebra $B$ generated by the $\lambda(t)$'s, agrees with the commutative von Neumann algebra generated by $h$. Now for any $\epsilon>0$, 
$(\epsilon \I + h)^{-1/2} \equiv (\epsilon + h)^{-1/2}$ is bounded, and so belongs to $B$.  So it is a simple matter to use parts (1) and (4) of Theorem \ref{Thm1} to see that 
\begin{eqnarray*}
\widehat{\nu}((\epsilon + h)^{-1/2}\widetilde{T}(a)(\epsilon + h)^{-1/2}) &=& \widehat{\nu}(\widetilde{T}((\epsilon + h)^{-1/2}a(\epsilon + h)^{-1/2}))\\
&\leq& \widehat{\nu}((\epsilon + h)^{-1/2}a(\epsilon + h)^{-1/2}).
\end{eqnarray*}
Letting $\epsilon$ decrease to zero, now yields $\tau(\widetilde{T}(a))\leq \tau$, as required.
\end{proof}

With the above corollary at our disposal we may at this point appeal to Yeadon's ergodic theorem for positive maps \cite{yea} (recently extended and significantly sharpened by Haagerup, Junge and Xu \cite[Theorem 5.1]{HJX}) to see that $\widetilde{T}$ extends canonically to a bounded map on $L^1(\cM, \tau)$. Since $\cM$ is semifinite, the full power of the DDdP approach is therefore at our disposal, and we may use real interpolation (see \cite{DDdP2}) to extend the action of $\widetilde{T}$ to a large class of rearrangement invariant Banach function spaces associated with $\cM$. Specifically $\widetilde{T}$ canonically induces an action on each Orlicz space $L^\Phi(\cM, \tau)$. To highlight the importance of this observation we provide some details in:
\begin{remark}
Dodds, Dodds, de Pagter have shown (see Theorem 3.2 in \cite{DDdP2} and Remark \ref{5.3}) that the classical interpolation scheme for fully symmetric Banach function spaces on $\Rn^+$ can be quantized in terms of $\cM$. In particular, this implies that well defined classical dynamical maps on classical Orlicz spaces can be quantized. In other words, there is a powerful recipe for defining the large family of quantum maps on quantum Orlicz spaces associated with $\cM$.
\end{remark}

Whilst this fact is worthy of noting, here we are interested in the action of $\widetilde{T}$ on spaces more associated with 
$\mathfrak{M}$, not $\cM$. For this purpose we need the following observation:

\begin{proposition} Let $T$ be a completely positive Markov map on $\mathfrak{M}$ satisfying $\nu\circ T\leq \nu$ and let 
$\widetilde{T}$ be its completely positive extension to $\mathcal{M} = \mathfrak{M}\rtimes_{\sigma^{\nu}}{\mathbb R}$. Then $\widetilde{T}$ canonically induces a map on the space $(L^\infty+L^1)(\mathcal{M},\tau)$. 
\end{proposition}

\begin{proof} If we apply \cite[Theorem 5.1]{HJX} to Corollary \ref{Ttrace}, it is clear that $\widetilde{T}$ canonically induces a map on $L^1(\mathcal{M},\tau)$. So this claim follows from the theory of Dodds, Dodds and de Pagter \cite{DDdP2}. 
\end{proof}

\begin{proposition}\label{criterion} Let $T$ be a completely positive Markov map on $\mathfrak{M}$ satisfying $\nu\circ T\leq \nu$ and let $\widetilde{T}$ be its completely positive extension to $\mathcal{M} = \mathfrak{M}\rtimes_{\sigma^{\nu}}{\mathbb R}$. For the sake of simplicity we will also write $\widetilde{T}$ for the extension to $(L^\infty+L^1)(\mathcal{M},\tau)$. In its action on $(L^\infty+L^1)(\mathcal{M},\tau)$, it satisfies the condition that $\widetilde{T}(ab)=a\widetilde{T}(b)$ for all $a\in B$ and all $b\in(L^\infty+L^1)(\mathcal{M},\tau)$. (Here $B$ is von Neumann subalgebra generated by all $\lambda(s)$, $s\in \mathbb{R}$.)
\end{proposition}

\begin{proof} The stated property is known to hold for $\mathcal{M}$. Since $\mathcal{M}\cap L^1(\mathcal{M},\tau_\mathcal{M})$ is norm-dense in $L^1(\mathcal{M},\tau_\mathcal{M})$, the continuity of $\widetilde{T}$ on $L^1(\mathcal{M},\tau_\mathcal{M})$, then ensures that it also holds for $L^1(\mathcal{M},\tau_\mathcal{M})$.
\end{proof}

The reason for proving the above, is that all the type III Orlicz spaces with upper fundamental index less that 1, live inside $(L^\infty+L^1)(\mathcal{M},\tau)$. We first proceed to define the fundamental indices of an Orlicz space. These were introduced by Zippin \cite{Zip}.

\begin{definition} Let $L^\Psi(\mathfrak{M})$ be an Orlicz space, and let $\varphi_\psi$ be the fundamental function of the space 
$L^\Psi(0,\infty)$. (Here we consider the Luxemburg, norm but the case of the Orlicz norm is completely analogous.) Let $M_\psi(t)=\sup_{s>0}\frac{\varphi_\Psi(st)}{\varphi_\Psi(s)}$. (If we use the fundamental function of the space $L_\Psi(0,\infty)$ equipped with the Orlicz norm, we will write $\widetilde{M}_\psi(t)$ for this function.) Then the lower and upper fundamental indices of $L^\Psi(\mathfrak{M})$ are defined to be $$\underline{\beta}_{L^\Psi}=\sup_{0<t\leq 1}\frac{\log M_\psi(s)}{\log s}\quad \mbox{and}\quad \overline{\beta}_{L^\Psi}=\inf_{1<t}\frac{\log M_\psi(s)}{\log s}$$respectively. 
\end{definition}

The following result is a variant of \cite[Theorem 3.13]{L}, where the Boyd indices were used to prove a similar theorem. 

\begin{proposition}Let $L^\Psi(\mathfrak{M})$ be an Orlicz space with upper fundamental index strictly less than 1. Then $L^\psi(\mathfrak{M})\subset (L^\infty + L^1)(\mathcal{M},\tau_\mathcal{M})$ where $\mathcal{M} = \mathfrak{M}\rtimes_{\sigma^{\nu}}{\mathbb R}$. Moreover the canonical topology on $L^\psi(\mathfrak{M})$ then agrees with the subspace topology inherited from $(L^\infty + L^1)(\mathcal{M},\tau_\mathcal{M})$. 
\end{proposition}

\begin{proof} It is clear from exercise 14 of \cite[Chapter 3]{BS} (see also Corollary 8.15, page 275 in \cite{BS}) that $0<\underline{\beta}_{\Psi^*} =1-\overline{\beta}_\Psi= \sup_{0<t\leq 1} \frac{\log \widetilde{M}_{\psi^*(t)}}{\log t} = \lim_{t\to 0^+}\frac{\log \widetilde{M}_{\psi^*(t)}}{\log t}$. So given any $n\in \Nn$, there must exist $1\geq t_0>0$ such that $$\frac{\log(\widetilde{M}_{\psi^*(t)})}{\log t}\geq \frac{n-1}{n}\underline{\beta}_{\Psi^*}$$for all $0<t\leq t_0$. This can be shown to imply the fact that $$\widetilde{M}_{\psi^*(t)}\leq t^{\frac{n-1}{n}\underline{\beta}_{\Psi^*}}.$$ 

It is an exercise to see that  $$\frac{\widetilde{\varphi}_{\Psi^*(s)}}{\widetilde{\varphi}_{\Psi^*(s/t)}}\leq \sup_{r>0}\frac{\widetilde{\varphi}_{\Psi^*(rt)}}{\widetilde{\varphi}_{\Psi^*(r)}}=\widetilde{M}_{\psi^*(t)}.$$On selecting $s_t<0$ so that 
$e^{s_t}=t$, this in turn ensures that $$d_{s_t}=\widetilde{\varphi}_{\Psi^*}(e^{-s_t}h)^{-1}\widetilde{\varphi}_{\Psi^*}(h)\leq t^{\frac{n-1}{n}\underline{\beta}_{\Psi^*}}.$$Let $x\in L^\Psi(\mathfrak{M})$ be given. We remind the reader that then $\mu_1(x)=\sup_{0<t\leq 1}t\mu_t(x)$ and $t\mu_t(x)=\mu_1(d_{s_t}^{1/2}xd_{s_t}^{1/2})$ for any $0<t\leq 1$ (see \cite[Theorem 3.10]{L} and its proof). For 
$0<t\leq t_0$ we then clearly have that $t\mu_t(x)\leq\|d_{s_t}\|\mu_1(a)\leq t^{\frac{n-1}{n}\overline{\beta}_\Psi}$. Finally use the 
fact that $t\to \mu_t(x)$ is decreasing, to see that for $n>1$
\begin{eqnarray*}
\mu_1(x)&\leq& \int_0^1\mu_r(x)\,dr=\int_0^1\frac{r\mu_r(x)}{r}\,dr\\
&\leq& \left[\int_0^{t_0}\left(\frac{1}{r}\right)^{1-\frac{n-1}{n}\underline{\beta}_{\Psi^*}}\,dt + \int_{t_0}^1\frac{1}{r}\,dr\right]\sup_{0<t\leq 1}t\mu_t(x).
\end{eqnarray*}
Hence $\mu_1(x)=\sup_{0<t\leq 1}t\mu_t(x)$ is equivalent to the canonical norm on $(L^\infty + L^1)(\mathcal{M},\tau_\mathcal{M})$, namely 
$\int_0^1\mu_r(x)\,dr$. The claim follows
\end{proof}

What we still need is an alternative criterion for identifying the elements of $(L^\infty+L^1)(\mathcal{M},\tau)$ that belong to some $L^\Psi(\mathfrak{M})$

\begin{lemma}\label{Orlicz} Let $\mathcal{M} = \mathfrak{M}\rtimes_{\sigma^{\varphi}}{\mathbb R}$, let $\theta_s$ be the dual action of $\mathbb{R}$ on $\mathcal{M}$, and let $h=\frac{d\widehat{\varphi}}{d\tau}$. Given some Young's function $\Psi$ and some $s\leq 0$, we will write $d_s$ for the bounded map $\widetilde{\varphi}_{\Psi^*}(e^{-s}h)^{-1}\widetilde{\varphi}_{\Psi^*}(h)$. We then have that $a\in \widetilde{\mathcal{M}}$ belongs to $L^\Psi(\mathfrak{M})$ if and only if for any $s\leq 0$ we have that $\theta_s(a)=e^{-s}d_s^{1/2}ad_s^{1/2}$.
\end{lemma}

\begin{proof} The ``only if'' part was proven in the proof of \cite[3.10]{L}. For the converse assume that the stated condition regarding the action of $\theta_s$ holds for some $a\in \widetilde{\mathcal{M}}$. Note that the action of the $\theta$'s extends to operators affiliated to $\mathcal{M}$. So for any projection $e\in \mathfrak{M}$ of finite weight, we know that $e.\widetilde{\varphi}_{\Psi^*}(h)$ is closable with $\tau_\mathcal{M}$-dense domain. It is easy to conclude that
\begin{eqnarray*}
\theta_s(e[\widetilde{\varphi}_{\Psi^*}(h)\chi_{[0,n]}(h)])&=&\theta_s(e)\theta_s(\widetilde{\varphi}_{\Psi^*}(h)\chi_{[0,n]}(h)\\
&=& e\widetilde{\varphi}_{\Psi^*}(e^{-s}h)\chi_{[0,n]}(e^{-s}h)\\
&=& e\widetilde{\varphi}_{\Psi^*}(e^{-s}h)\chi_{[0,e^sn]}(h).
\end{eqnarray*}
But then the operators $\theta_s([e\widetilde{\varphi}_{\Psi^*}(h)])$ and $[e.\widetilde{\varphi}_{\Psi^*}(e^{-s}h)]$ must agree on the dense subspace $\cup_{\lambda<\infty}\chi_{[0,\lambda](h)}(\mathfrak{H})$. Hence $[e.\theta_s(\widetilde{\varphi}_{\Psi^*}(h)]=[e.\widetilde{\varphi}_{\Psi^*}(e^{-s}h)]$ \cite[Lemma 2.1]{Gol}. By duality $\theta_s(\widetilde{\varphi}_{\Psi^*}(h)e)=\widetilde{\varphi}_{\Psi^*}(e^{-s}h)e$.  Hence $d_s^{1/2}\theta_s(\widetilde{\varphi}_{\Psi^*}(h)e)=\widetilde{\varphi}_{\Psi^*}(h)e$. By duality we then also have that $\theta_s([e\widetilde{\varphi}_{\Psi^*}(h)])d_s^{1/2}=[e\widetilde{\varphi}_{\Psi^*}(h)]$. Consequently for any two projections $e,f\in \mathfrak{M}$ of finite weight, we have that 
\begin{eqnarray*}
\theta_s([f\widetilde{\varphi}_{\Psi^*}(h)]a(\widetilde{\varphi}_{\Psi^*}(h)e))&=&\theta_s([f\widetilde{\varphi}_{\Psi^*}(h)])\theta_s(a)\theta_s(\widetilde{\varphi}_{\Psi^*}(h)e)\\
&=&e^{-s}\theta_s([f\widetilde{\varphi}_{\Psi^*}(h)])d_s^{1/2}ad_s^{1/2}\theta_s(\widetilde{\varphi}_{\Psi^*}(h)e) \\
&=&e^{-s}[f\widetilde{\varphi}_{\Psi^*}(h)]a(\widetilde{\varphi}_{\Psi^*}(h)e).
\end{eqnarray*}
for all $s\leq 0$. It is now an exercise to see that if for some $a_0\in \widetilde{\mathcal{M}}$ we have that $\theta_s(a_0)=e^{-s}a_0$ for all $s\leq0$, then $\theta_s(a_0)=e^{-s}a_0$ for all $s\in \mathbb{R}$. So we have that $\theta_s([f\widetilde{\varphi}_{\Psi^*}(h)]a(\widetilde{\varphi}_{\Psi^*}(h)e))=e^{-s}[f\widetilde{\varphi}_{\Psi^*}(h)]a(\widetilde{\varphi}_{\Psi^*}(h)e)$ for all $s\in \mathbb{R}$. Hence by definition $a\in L^\Psi(\mathfrak{M})$ (see \cite{L}).
\end{proof}

\begin{theorem}\label{orlthm} Let $\mathcal{M} = \mathfrak{M}\rtimes_{\sigma^{\nu}}{\mathbb R}$, and let $\Psi$ be a Young's function. Let $\widetilde{T}$ be the map induced by $T$ on $(L^\infty+L^1)(\mathcal{M},\tau)$. If $\overline{\beta}_{L^\Psi}<1$, then $\widetilde{T}$ restricts to a bounded map on $L^\Psi(\mathfrak{M})$.
\end{theorem}

\begin{proof} If $\overline{\beta}_{L^\Psi}<1$, then $L^\Psi(\mathfrak{M})$ lives inside $(L^\infty+L^1)(\mathcal{M},\tau)$, and also gets its topology from $(L^\infty+L^1)(\mathcal{M},\tau)$. So all one needs to do to prove this, is to note  
\begin{itemize}
\item that if $a \in (L^\infty+L^1)(\mathcal{M},\tau)$, then also $d_t^{1/2}ad_t^{1/2}\in (L^\infty+L^1)(\mathcal{M},\tau)$,
\item then that $\widetilde{T}(d_t^{1/2}ad_t^{1/2})=d_t^{1/2}\widetilde{T}(a)d_t^{1/2}$ by Proposition \ref{criterion}, 
\item and then simply apply the Lemma.
\end{itemize}
\end{proof}

\begin{example} The upper fundamental index of the space $L^{\cosh -1}(0,\infty)$ is $\frac{1}{2}$. Isomorphic Orlicz spaces share the same indices. So it is sufficient to prove this for a space isomorphic to $L^{\cosh -1}(0,\infty)$. We show how to construct such space before proving the claim. It is easy to see that the graphs of $e^t$ and $\frac{e^2}{4}t^2$ are tangent at $t=2$. This fact ensures that
$$\Psi_e(t)=\left\{\begin{array}{lll} \tfrac{e^2}{4}t^2 & \mbox{if} & 0\leq t \leq 2\\ e^{t} & \mbox{if} & 2 < t \end{array}\right.$$is a Young's function. Using Maclaurin series it is easy to see that $\lim_{t\to 0+}\frac{\cosh(t)-1}{\Psi_e(t)}=\frac{2}{e^2}$. Since we also have that $\lim_{t\to \infty}\frac{\cosh(t)-1}{\Psi_e(t)}=\lim_{t\to \infty}\frac{e^t+e^{-t}-2}{2e^t}=\frac{1}{2}$, it is clear that $\Psi_e\approx \cosh-1$. (To see this note that the limit formulae ensure that we may find $0< \alpha<\beta<\infty$ so that $\frac{1}{e^2} < \frac{\cosh(t)-1}{\Psi_e(t)} < \frac{3}{e^2}$ on $[0,\alpha]$, and $\frac{1}{4} < \frac{\cosh(t)-1}{\Psi_e(t)} < \frac{3}{4}$ on $[\beta,\infty)$. Since the function $\frac{\cosh(t)-1}{\Psi_e(t)}$ has a both a minimum and maximum on the interval $[\alpha,\beta]$, a combination of these facts ensures that we can find positive constants $0<m<M<\infty$ so that 
$m \Psi_e(t) < \cosh(t)-1 < M\Psi_e(t)$ for all $t\in [0,\infty)$.) This is clearly enough to ensure that $L^{\Psi_e}(0,\infty)\equiv L^{\cosh-1}(0,\infty)$. 

It remains to compute the fundamental indices of $L^{\Psi_e}(0,\infty)$. We will use the formulas in Remark 2.3 of \cite{Zip} to compute these indices. We will assume that $L^{\Psi_e}(0,\infty)$ is equipped with the Luxemburg norm. Since 
$$\Psi_e^{-1}(t)=\left\{\begin{array}{lll} \tfrac{2}{e}t^{1/2} & \mbox{if} & 0\leq t \leq e^2\\ \log(t) & \mbox{if} & e^2 < t \end{array}\right.,$$it now follows from \cite[4.8.17]{BS} that the fundamental function of $L^{\Psi_e}(0,\infty)$ is given by 
$$\varphi_e(t)=\left\{\begin{array}{lll} \tfrac{e}{2}t^{1/2} & \mbox{if} &  t \geq e^{-2}\\ \tfrac{1}{-\log(t)} & \mbox{if} & t < e^{-2} \end{array}\right.$$
We proceed to compute the function $M_{\Psi_e}(s)=\sup_{t>0}\frac{\varphi_e(st)}{\varphi_e(t)}$. In computing this function, we first consider the case where $0<s\leq 1$. Since $\varphi_e$ is increasing, we then have that
$$\frac{\varphi_e(st)}{\varphi_e(t)}\leq \frac{\varphi_e(t)}{\varphi_e(t)}=1$$ for any $t>0$. Since we also have that
$$\lim_{t\to 0}\frac{\varphi_e(st)}{\varphi_e(t)}=\lim_{t\to 0} \frac{\log(t)}{\log(st)}=\lim_{t\to 0} \frac{\log(t)}{\log(s)+\log(t)}=1,$$it is clear that $M_{\Psi_e}(s)=1$ in this case.

Now let $s$ be given with $s > 1$. We then have that 

$$\frac{\varphi_e(st)}{\varphi_e(t)} =\left\{\begin{array}{lll} s^{1/2} & \mbox{if} &  t > \tfrac{1}{e^2}\\ -\tfrac{e}{2}s^{1/2}t^{1/2}\log(t) & \mbox{if} & \tfrac{1}{e^2} > t > \tfrac{1}{se^2}\\ \tfrac{\log(t)}{\log(st)} & \mbox{if} & t < \tfrac{1}{se^2}  \end{array}\right. $$

It is not too difficult to see that the function $t\to-\tfrac{e}{2}s^{1/2}t^{1/2}\log(t)$ has a maximum of $s^{1/2}$ at $t=e^{-2}$ on the interval $(0,1)$. So for $t\in(\tfrac{1}{se^2}, \tfrac{1}{e^2})$, the supremum of the above quotient is $s^{1/2}$. Finally consider the function $$t\to\tfrac{\log(t)}{\log(st)}=\tfrac{\log(t)}{\log(s)+\log(t)} =1-\tfrac{\log(s)}{\log(s)+\log(t)}=1-\tfrac{\log(s)}{\log(st)}.$$It is easy to see that
$$\frac{d}{dt}(1-\tfrac{\log(s)}{\log(s)+\log(t)})=\frac{\log(s)}{t(\log(st))^2} >0$$ on $t\in(0,\tfrac{1}{se^2})$. Hence on $(0,\tfrac{1}{se^2}]$, $t\to \tfrac{\log(t)}{\log(st)}$ attains a maximum of $1+\tfrac{1}{2}\log(s)$ at $t=\tfrac{1}{se^2}$. Using the fact that $1+\log(t)\leq t$, it is now easy to see that $1+\tfrac{1}{2}\log(s)=1+\log(s^{1/2})\leq s^{1/2}$. Putting all these facts together leads to the conclusion that $M_{\Psi_e}(s)=\sup_{t>0}\frac{\varphi_e(st)}{\varphi_e(t)}=s^{1/2}$ in this case. We therefore have that $$\overline{\beta}_{\Psi_e}=\lim_{s\to\infty}\frac{\log M_{\Psi_e}(s)}{\log s}=\lim_{s\to\infty}\frac{\log s^{1/2}}{\log s}=\frac{1}{2}$$as claimed. Similarly $$\underline{\beta}_{\Psi_e}=\lim_{s\to 0+}\frac{\log M_{\Psi_e}(s)}{\log s}=\lim_{s\to 0+}\frac{\log 1}{\log s}=0.$$
\end{example}

\begin{corollary}
If $T$ is a CP map on $\mathfrak{M}$ satisfying DBC, then $\widetilde{T}$ canonically induces an action on $L^{\cosh -1}(\mathfrak{M})$.
\end{corollary}

We pause to consider the case where $T$ is a positive normal map on $\mathfrak{M}$ which may not be CP. For the sake of simplicity assume that $\mathfrak{M}$ is $\sigma$-finite, with $\omega$ a faithful normal state. Here we will follow the exposition of \cite{GL},\cite {GL2}, \cite{Kos}, and \cite{terp}. There is an operator $h \in L^1(\mathfrak{M})$ (where $\widetilde{\omega}$ is the dual weight of $\omega$ and $ h \equiv h_{\omega} \equiv \frac{d\widetilde{\omega}}{d\tau}$) associated with the state $\omega$. One may then use this operator to define embeddings
of $\mathfrak{M}$ into $L^p(\mathfrak{M})$ as follows:
\begin{equation}
\iota_p : \mathfrak{M} \ni a \mapsto h_{\omega}^{\frac{1}{2p}} a h_{\omega}^{\frac{1}{2p}}.
\end{equation}
Further, let $T$ be a Markov map on $\mathfrak{M}$. Define $T^{(p)}$ by
\begin{equation}
T^{(p)}(\iota_p(a)) = \iota_p(Ta),
\end{equation}
for $a \in \mathfrak{M}$. Note that $\iota_p(\mathfrak{M})$ is dense in $L^p(\mathfrak{M})$ (see Lemma 1.6 in \cite{GL}). Therefore, $T^{(p)}$ is densely defined. (Similar conclusions hold in the general non-$\sigma$-finite case where we have a weight $\nu$ instead of a state $\omega$. But in that case the embedding $\iota_p$ should be defined on the subalgebra $\mathrm{span}\{y^*x: x, y \in \mathfrak{M}; \nu(x^*x)<\infty, \nu(y^*y)<\infty\}$ rather than the full algebra.)

Moreover by Theorem 5.1 in \cite{HJX} (see also \cite[Proposition 2.2]{GL}), whenever $\nu\circ T\leq \gamma \nu$ for some $\gamma>0$, the map $T^{(p)}$ will extend to a positive bounded map on $L^p(\mathfrak{M})$.

Of course the question now arises as to how the maps $T^{(p)}$ compare to the extension of $\tT$ to $L^p(\mathfrak{M})$ by means of the above process, and ultimately also how the work of Goldstein and Lindsay, and Haagerup, Junge and Xu, compare to ours. This relationship is clarified by the following corollary to Theorem \ref{orlthm}:

\begin{corollary}\label{orlcor} In the case $\Psi(t)=t^p$ ($p>1$), the maps induced by $\widetilde{T}$ on $L^p(\mathfrak{M})$, are exactly the maps $T^{(p)}$ constructed in \cite[Theorem 5.1]{HJX}.
\end{corollary}

\begin{proof} Firstly note that by the preceding theorem, each $L^p(\mathfrak{M})$ lives inside $(L^1+L^\infty)(\mathcal{M})$. Hence let $1<p<\infty$ and  $a\in \mathfrak{M}$ be given. Then $h^{1/(2p)}ah^{1/(2p)}\in L^p(\mathfrak{M})\subset (L^1+L^\infty)(\mathcal{M})$. On applying Proposition \ref{criterion} to $\widetilde{T}$, it follows that $\chi_{[0,n]}(h)\widetilde{T}(h^{1/(2p)}ah^{1/(2p)})\chi_{[0,n]}(h)=\widetilde{T}([h^{1/(2p)}\chi_{[0,n]}(h)]a[h^{1/(2p)}\chi_{[0,n]}(h)])$ for each $n\in \mathbb{N}$. But since $[h^{1/(2p)}\chi_{[0,n]}(h)]\in B$, it follows from the definition of $\widetilde{T}$ that \begin{eqnarray*}
\widetilde{T}([h^{1/(2p)}\chi_{[0,n]}(h)]a[h^{1/(2p)}\chi_{[0,n]}(h)]) &=& [h^{1/(2p)}\chi_{[0,n]}(h)]T(a)[h^{1/(2p)}\chi_{[0,n]}(h)]\\
&=& \chi_{[0,n]}(h)[[h^{1/(2p)}T(a)h^{1/(2p)}]\chi_{[0,n]}(h).
\end{eqnarray*}Hence for each $n$, we have 
\begin{eqnarray*}\chi_{[0,n]}(h)\widetilde{T}(h^{1/(2p)}ah^{1/(2p)})\chi_{[0,n]}(h) &=& \chi_{[0,n]}(h)[h^{1/(2p)}T(a)h^{1/(2p)}]\chi_{[0,n]}(h)\\
&=& \chi_{[0,n]}(h)T^{(p)}(h^{1/(2p)}ah^{1/(2p)})\chi_{[0,n]}(h).
\end{eqnarray*}
This is enough to prove the claim.
\end{proof}

One may view the extension of the map $T^{(p)}$ to $L^p(\mathfrak{M})$, as an extension of the Schr\"odinger type of evolution for the non-commutative $L^p(\mathfrak{M})$-space. To justify this assertion, let us consider $T^{(2)}$ in detail. Assume DBC for $T$ and denote by $\mathfrak{M}_0$ the $^*$-algebra of entire analytic elements for the modular automorphism $\sigma$  (we refer the reader to Section 2.5.3 in \cite{BR} for a description of $\mathfrak{M}_0$) and note that for $a \in \mathfrak{M}_0$, 
\begin{equation}
T^{(2)}(\iota_2(a)) = T(h^{\frac{1}{4}} a h^{\frac{1}{4}})
= h^{\frac{1}{4}} T(a) h^{-\frac{1}{4}} h^{\frac{1}{2}}
= \sigma_{- \frac{i}{4}}(T(a)) h^{\frac{1}{2}} = T(\sigma_{-\frac{i}{4}}(a)) h^{\frac{1}{2}},
\end{equation}
where we have used the commutativity of $T$ with the modular automorphism (ensured by the DBC assumption), and the fact that  $T^{(2)}$ is well defined map on $L^2(\mathfrak{M})$-space.
On the other hand, a Markov map $T$ on $\mathfrak{M}$ satisfying DBC leads to the following definition (cf. (\ref{for1.12})):
\begin{equation}
\widehat{T} x \Omega = T(x) \Omega.
\end{equation}
where $\widehat{T}$ leaves the natural cone $\cP$ (globally) invariant.
But, because of the universality of the standard form (see \cite{Araki}, \cite{Haage}, \cite{terp}) and hence also of the natural cone, we may identify
$$(\mathfrak{M},\ L^2(\mathfrak{M}),\ ^*, \ L^2_+(\mathfrak{M})) $$
with
$$(\mathfrak{M},\ \cH, \ J, \ \cP).$$
The unique implementing vector for $\omega$ (fixed, faithful state) in the natural cone $L^2_+(\mathfrak{M})$, is exactly $h^{\frac{1}{2}}$ while $\Omega$ plays the same role in $\cP$. As $\widehat{T}$ can be called a quantum map in the Schr\"odinger picture, it should now be obvious, that the maps $T^{(p)}$ considered in this subsection can be regarded as Schr\"odinger type quantum maps.

\section{Dirichlet forms and quantum maps on noncommutative spaces}

For the reader's convenience we recall that in late $50$'s of the last century, A. Beurling and J. Deny \cite{BD1}, \cite{BD2} 
characterized Markovian semigroups on Hilbert spaces of square integrable functions, in terms of Dirichlet forms associated to their infinitesimal generators. The reader is advised to consult Ma and R\"ockner's book \cite{MAR} for more details and a complete bibliography.

Subsequently, the theory of non-commutative Dirichlet forms was started with seminal papers of L. Gross \cite{gros} and S. Albeverio, R. H{\o}gh-Krohn \cite{AHK}. Then, the theory was elaborated by others; see Fabio Cipriani's thoroughgoing review \cite{Cip} for a recent account of the theory and a comprehensive bibliography.

In particular, it is worth pointing out, see \cite{MZ1}, \cite{MZ2}, that techniques based on non-commutative $L^p$-spaces were essential for the quantization of Markov-Feller processes.
Although the obtained framework seems to be the well adapted to the quantization of Markov-Feller processes, there emerged  problems 
regarding how one may correctly describe the stability of a system, a ``return to equilibrium'', and to provide a characterization of a ground state.

 One may conjecture that to answer the above questions one should generalize quantum $L^p$-techniques to those based on quantum Orlicz spaces. To support this conjecture, we emphasize that a basic ingredient in the examination of the above problems, are the log Sobolev inequalities (see \cite{Zeg} and references given there). The important point to note here is that Gross in his seminal paper \cite{gross} has already recognized the relevance of Orlicz spaces for the analysis of log Sobolev inequalities. In particular log Sobolev inequalities can be viewed as Poincar\'e-type inequalities in the Orlicz space $L\log(L+1)$; see \cite{Bob}. This is precisely what one may expect, as the Orlicz space $L\log(L+1)$ (up to an equivalent renorming dual to $L^{\cosh - 1}$) describes states with a nicely defined entropy function; for details see \cite{ML1}. Furthermore, it is worth pointing out that Gross \cite{gross} had already observed the crucial role of Dirichlet forms in the definition and examination of log Sobolev inequalities (see also \cite{Zeg}).
 
 Having such a strong motivation, this section will be devoted to a study of quantum maps on the Orlicz space of regular observables, which arise from maps on $L^2(\mathfrak{M})$-space or from Dirichlet forms.

As was assumed in the bulk of this paper (with the exception of \S \ref{1.4s}), $\mathfrak{M}$ will be a $\sigma$-finite von Neumann algebra acting on a Hilbert space $\cH$, $\cM$ the cross-product $\mathfrak{M} \rtimes_{\sigma} \Rn$, etc.

The basic Hilbert space, which will be used throughout this section, is $\cH$.
However, it is important to remember that one can identify (cf section \ref{1.4s})
\begin{equation}
 (\ \mathfrak{M}, \ L^2(\mathfrak{M}), \ ^*, \ L^2_+(\mathfrak{M}), \  h_{\omega}^{\frac{1}{2}} \ )
\end{equation}
with
\begin{equation}
 ( \ \mathfrak{M}, \ \cH, \  J, \ \cP, \ \Omega \ ), 
\end{equation}
where we emphasize that $J$ stands for modular conjugation!

We will need (see \cite{Glaz}) 
\begin{definition}
Let $J_0: \cH \to \cH$ be a conjugation. 
A linear operator $A$ with a domain of definition $D(A)$ dense in a complex Hilbert space $\cH$ is said to be $J_0$-self-adjoint if 
$$J_0 \ A \ J_0 = A^*.$$
\end{definition}

From now on we make the assumption that dynamical maps satisfy DBC and $J_0$ is the conjugation induced by the reversing operation $\Theta$, i.e. $J_0 x \Omega = \Theta(x) \Omega$ , $x \in \mathfrak{M}$; cf Section \ref{DBCo}.
It is important to note here that \textit{the above conjugation $J_0$, is in general different from the modular conjugation $J$}. To see the difference it is enough to note that for the modular conjugation $J$ one has $J \cP = \cP$ while DBC needs the conjugation $J_0$ with the property $J_0 \mathfrak{M}_+\Omega \subseteq \rm{closure}\{\mathfrak{M}_+ \Omega\}$. In other words, $J$ is associated with the natural cone $\cP$
whilst $J_0$ is more related to the cone $V_0$ (in Araki's terminology; cf \cite{Araki}), where $V_0 = \rm{closure}\{ \mathfrak{M}_+\Omega\}$.

From now on, $\{ \widehat{T}^0_t: t\geq 0 \}$ denotes a strongly continuous, $J_0$-self-adjoint semigroup on a Hilbert space $\cH$ in standard form. Let $(L, D(L))$ denote its $J_0$-self-adjoint infinitesimal generator and $(\cE, \cF)$ the associated closed quadratic form, i.e.
$$ \cE: \cH \to \Cn, \quad \rm{so \ that} \quad \cF = \{ \xi \in \cH: |\cE[\xi]| < + \infty \},$$
and

$$\cE[\xi] = (\xi, L \xi).$$  
The quadratic form $\cE$ is said to be $J$-real if
\begin{equation}
\label{6.1}
\cE[J \xi] = \cE[\xi], \quad \xi \in \cH.
\end{equation}

The important point to note here is that we have at our disposal two conjugations: $J$ and $J_0$. $J_0$ as used in the definition of $J_0$-self-adjointness, arises from choosing an arbitrary but fixed 
dynamical map $T: \mathfrak{M} \to \mathfrak{M}$ which satisfies DBC. The second map $J$, stems from the standard form of the von Neumann algebra $\mathfrak{M}$ and its faithful normal state $\omega$.
We need both, and both will be used. 

The necessity for the employment of both conjugations, can be seen from the following observations: One cannot do without the modular conjugation $J$, since it describes the canonical real Hilbert structure associated with the standard form of the von
Neumann algebra $\mathfrak{M}$. However $J_0$ is also needed since it gives the reversibility of the Hamiltonian part of dynamics whilst $J$, due to the fact $J\Delta J = \Delta^{-1}$ ($\Delta$ the modular operator) is not even able to implement the reversibility of the modular automorphism group. A description of reversibility of $\sigma_t(\cdot) = \Delta^{it} \cdot \Delta^{-it}$ is of paramount importance in applications to quantum statistical mechanics, as the equilibrium dynamics is usually described by a modular group.

\subsection{Dirichlet forms and Markovian maps}
\label{s6.1}
The mentioned (classical) Beurling-Deny characterization of Markovian semigroups 
relies on the following conditions:
\begin{enumerate}
\item $\cH_{\Rn}$ is assumed to be a real Hilbert space. $\cH_{\Rn}$ can be taken to be 
\begin{equation}
\cH^J = \{ \xi \in \cH; J \xi = \xi \}.
\end{equation}
However, we note that for a full analysis of a semigroup $T_t$, it seems to be standard procedure to make use of the complexification of the real Hilbert space (see pp 23-25 in \cite{MAR}).
\item 
\label{hoho}$\cE$ is real valued, i.e. $\cE: \cH_{\Rn} \to (\infty, \infty]$.
\item
Quadratic forms which do not increase their values under the map $u \mapsto u\wedge1$, i.e.
\begin{equation}
\label{1.51}
\cE[u\wedge1] \leq \cE[u].
\end{equation}

Here, $u\wedge v$ is defined to be the projection of the vector $u$ (from the real Hilbert space) onto the cone $v + \cP$; for details see \cite{Cip}, Section 1.12 in \cite{edwards} and/or \cite{IM}.

\end{enumerate}

\begin{remark}
\label{6.2}
Condition \ref{hoho} seems to be a very strong one and needs some elaboration. To this end, we recall that in Hilbert space terms, the general form of an infinitesimal generator $L$ of uniformly continuous quantum dynamical semigroup $\widehat{T}$ is (see \cite{kos}, \cite{lind}, \cite{chris})
\begin{equation}
L = i H - D,
\end{equation}
where $H = H^*$ and $D\geq 0$. 
Thus, the condition  (\ref{hoho}) demands the self-adjointness of $L$ on $\cH$. In particular, $\cE[\xi] = (\xi, (-D) \xi)$, where $\xi$ in the domain of $D$, is the most natural candidate for a Dirichlet form.
\textit{Consequently, the subsequent analysis given in this subsection, is reduced to a study of self-adjoint semigroups only!} It is easy to check that the requirement of selfadjointness of a dynamical map $\widehat{T}$ follows from the assumption that $\widehat{T}$ is $J_0$-real, i.e. $\widehat{T} \cH^{J_0} \subset \cH^{J_0}$, where $\cH^{J_0} = \{ \xi \in \cH; J_0 \xi = \xi \}$.
\end{remark}

In the remainder of this subsection we assume that the considered forms are bounded from below in the sense that, for some constant $c \in \Rn^+$
$$\cE[\xi] \geq c \|\xi\|^2, \quad \xi \in \cF.$$

The following definition is taken from Cipriani's paper (cf Definition 2.51 in \cite{Cip})

\begin{definition}
Let $( \mathfrak{M},\ \cH, \ \cP, J)$ be a standard form of a von Neumann algebra $\mathfrak{M}$ with $\Omega \in \cP$ a cyclic and separating vector.

A $J$-real quadratic form $\cE: \cH \to (-\infty, + \infty]$ is said to be Markovian with respect to $\Omega$ if
\begin{equation}
\cE[\xi \wedge \Omega] \leq \cE[\xi], \quad \forall \xi \in \cH^J.
\end{equation}
\textit{A closed Markovian form will be called a Dirichlet form.}
\end{definition}

Furthermore, we say that a map $\widehat{T}^0 : \cH \to \cH$ is Markovian if $\widehat{T}^0 \cP \subseteq \cP$ and $\widehat{T}^0 \Omega \leq \Omega$. We are now in a position to give
the following characterization of Markovian semigroups in terms of Dirichlet forms. Note that the given characterization is again a slight modification of Theorem 2.52 given in \cite{Cip}.

\begin{theorem}
Let $( \mathfrak{M},\ \cH, \ \cP, J)$ be a standard form of a von Neumann algebra $\mathfrak{M}$ with $\Omega \in \cP$ a cyclic and separating vector and $\{\widehat{T}_t^0; t\geq0 \}$ be $J_0$-real contractive, strongly continuous semigroup on the Hilbert space $\cH$ such that $\widehat{T}^0_t \cP \subseteq \cP$, where the conjugation $J_0$ is such that $J_0\Omega = \Omega$ and $J_0 \mathfrak{M}_+ \subseteq \overline{\mathfrak{M}_+ \Omega}$. Let $\cE: \cH \to [0, + \infty]$ be the associated $J$-real, closed quadratic form. The following properties are equivalent:
\begin{enumerate}
\item $\{\widehat{T}_t^0: t\geq0 \}$ is Markovian with respect to $\Omega$.
\item $\cE$ is a Dirichlet form with respect to $\Omega$.
\end{enumerate}
\end{theorem}
\begin{proof}
The proof is practically a repetition of Cipriani's arguments, see \cite{Cip}.
\end{proof}
\begin{remark}
\begin{enumerate}
\item It follows immediately that if $T: \mathfrak{M} \to \mathfrak{M}$ satisfies DBC then $\widehat{T}$ defined by $\widehat{T}x \Omega = T(x)\Omega$, $x \in \mathfrak{M}$, fulfills all requirements assumed for $\widehat{T}_0$.
\item Furthermore, if additionally $\widehat{T}$ is $J_0$-real then $\widehat{T}$ commutes with $J$. Consequently, it's infinitesimal generator gives a $J$-real form.
\end{enumerate}
\end{remark}
Thus we have arrived at the conclusion that Dirichlet forms defined on $\cH^J$ lead to self-adjoint ($J_0$-real) Markov semigroups on $\cH^J$. On the other hand (see subsection \ref{DBCo}), Markov maps on $\mathfrak{M}$ satisfying DBC are in one-to-one correspondence with $J_0$-selfadjoint semigroups $\widehat{T}$ on $\cH^J \subset \cH$, preserving (globally) the natural cone $\cP$, leaving the vector $\Omega$ invariant, and commuting with the modular automorphism. All this leads to the following question: \textit{Do these maps have their counterparts on the Orlicz space of regular observables?} This question gains an additional interest if we realize that the discussed Markov semigroups are nicely related to the corresponding semigroup defined on von Neumann algebras, see Section 2.4 in \cite{Cip}.
We will answer this question in the next subsection.

\subsection{Detailed Balance Condition and Orlicz spaces}
The answer to the question posed at the end of the previous subsection naturally falls into three parts.

The first one was already presented in Section \ref{1.4s}. In particular, employing two ways of quantising measurable structures (respectively the DDdP strategy and the strategy based on crossed products) quantum maps on the space of regular observables were defined, with the second of these strategies requiring the original map on the set of all bounded observables to at least be completely bounded.

For the second part we pause to note that the CP assumption can be dropped in certain cases. 
To see this, let a $\sigma$-finite von Neumann  algebra $\mathfrak{M}$ be given in the standard form
$$ ( \ \mathfrak{M}, \ \cH, \  J, \ \cP, \ \Omega \ ). $$
Let $T: \mathfrak{M} \to \mathfrak{M}$  be a Markov map satisfying DBC, so \textit{complete positivity is not assumed!}  
Neither does one need to require any sort of self-adjointness with respect to modular structures of the considered maps. In this context the map $T$ canonically induces an action on each of the $L^p(\mathfrak{M})$-spaces ($1\leq p <\infty$) by means of the maps $T^{(p)}$  discussed in section \ref{1.4s}. So in this context at least an action on $L^p$-spaces survives. As was pointed out in section \ref{1.4s}, in the case where we do have CP maps, our construction reproduces the action of these same maps on $L^p(\mathfrak{M})$-spaces. Hence our construction should be seen as an extension of the action of $T$ to a wider class of spaces in the CP case.

Finally, in the third part we follow the line of reasoning given in subsection \ref{s6.1}.
Let $T: \mathfrak{M} \to \mathfrak{M}$ be of the following Stinespring-Choi-Kraus form
\begin{equation}
T(x) = \sum_{i \in \Lambda} V^*_i x V_i
\end{equation}
where $V_i \in \mathfrak{M}$. If $\Lambda$ is infinite, the series $\sum_i V^*_i a V_i$ converges in the strong operator topology. As there is no loss of generality in assuming  $\sum_{i \in \Lambda}V^*_iV_i \leq \jed$, in the remainder of this subsection we require that this inequality holds.

\begin{remark}
Obviously in terms of the infinitesimal generator of a dynamical semigroup, being a CP map, the map $T$ can be considered as an archetype of dissipations. Therefore it gives an example, on the algebra level, of the part $D$ of the infinitesimal generator of a dynamical process, cf Remark \ref{6.2}.
On the the other hand, it is worth pointing out that the Hamiltonian form leads to an automorphism $\alpha$ on the original algebra $\mathfrak{M}$, i.e. $\alpha : \mathfrak{M} \to \mathfrak{M}$. Assuming that $\alpha$ commutes with the equilibrium dynamics, and so also with the modular automorphism, $\alpha$ has an extension on  $\cM \equiv \mathfrak{M} \rtimes_{\sigma} \Rn$, which is also a $^*$-homomorphism, cf. Theorem 4.1 in \cite{HJX} (loc. cit. as Theorem \ref{Thm1}). Then, the results given in
our previous papers, see Definition 4.2 in \cite{LM} and Proposition 4.7(i) in \cite{Louis}, yield a well defined map on quantum Orlicz spaces. As the full dynamics can be obtained by Trotter's product formula from 
the Hamiltonian and dissipative parts of time evolution, we may restrict ourselves to an analysis of the map $T$.
\end{remark}

Following the arguments given in the proof of Theorem 4.1 in \cite{HJX} it is easy to see that $T$ has an extension on $\cM \equiv \mathfrak{M} \rtimes_{\sigma} \Rn$ given by
\begin{equation}
\tT(\lambda(t)\pi(x)) = \sum_{i \in \Lambda} \widetilde{V_i}^* (\lambda(t) \pi(x) \widetilde{V_i}),
\end{equation}
where $\widetilde{V_i} \equiv V_i \otimes \jed_{L^2(\Rn)}$ and where we have used the notation given in Sections 3 and 4.

\emph{In the remainder of this discussion we will assume that $\Lambda$ is finite.} Suppose that $\Lambda$ can be identified with $N\in \mathbb{N}$. Given any $a \in {\cM}_+$, it then easily follows from \cite[Lemma 2.5]{FK}, that for any $t>0$, we will have that  $$\mu_t(\widetilde{T}(a))=\mu_t(\sum_{i=1}^N \widetilde{V}^*_i a \widetilde{V}_i)\leq (\sum_{i=1}^N\|\widetilde{V}_i\|^2)\mu_t(a)\leq N\mu_t(a).$$

For any two elements $a,b \in \cM$, we therefore have that $|\widetilde{T}(a)-\widetilde{T}(b)|^2\leq \widetilde{T}(|a-b|^2)$. Hence using what we have just shown, it then follows that $$\mu_t(\widetilde{T}(a)-\widetilde{T}(b))= \mu_t(|\widetilde{T}(a)-\widetilde{T}(b)|^2)^{1/2}\leq \mu_t(\widetilde{T}(|a-b|^2))^{1/2}$$ 
$$\leq N\mu_t(|a-b|^2)^{1/2}=N\mu_t(a-b).$$

Therefore in the case where $\Lambda$ is finite, the map $\widetilde{T}$ is in fact continuous on $\cM$ with respect to the topology of convergence in measure (\cite[Lemma 3.1]{FK}). Since $\cM$ is dense in $\widetilde{\cM}$ with respect to this topology, $\widetilde{T}$ then extends not just to $(L^1+L^\infty)(\cM, \tau)$ (as we saw in Theorem \ref{orlthm}), but to all of 
$\widetilde{\cM}$. Moreover in its action on $\widetilde{\cM}$, we have that $\sup_{0<t\leq 1}t\mu_t(\widetilde{T}(a))\leq N\sup_{0<t\leq 1}t\mu_t(a)$ for all $a\in \widetilde{\cM}$. Hence the extension of $\widetilde{T}$ restricts to a continuous map on the quasi-Banach function space $\Lambda_\infty(L^1+L^\infty)(\widetilde{\cM})$. (For details on this space and how it relates to the present study, refer to \cite[Remark 3.14]{L}.) 

What is of particular importance is that this quasi-Banach function space canonically contains all the quantum Orlicz spaces $L^\Phi(\mathfrak{M})$, and not just those satisfying the criteria in Theorem \ref{orlthm} (see \cite[Remark 3.14]{L}). On repeating the argument given in Theorem \ref{orlthm}, but using the space $\Lambda_\infty(L^1+L^\infty)(\widetilde{\cM})$ instead of $(L^1+L^\infty)(\cM,\tau)$, we may now show that this extension of $\widetilde{T}$ restricts to a bounded map on \emph{each} quantum Orlicz space $L^\Phi(\mathfrak{M})$.

To sum up this Section: a variety of dynamical maps, including those which are related to Dirichlet forms, canonically induce an action on Orlicz spaces which are distinguished by our approach to Statistical Physics, like for example $L^{\cosh -1}$ spaces.

\section{Conclusions and final remarks}
 By focusing attention on observables, algebras and states, we proposed a new formalism for Statistical Mechanics, both classical and quantum, see \cite{ML1}, \cite{ML2}.
It is based on two distinguished Orlicz spaces $L^{\cosh-1}$ and $L\log(L+1)$, and proves to be a canonical extension of the traditional formalism for elementary quantum mechanics; for details see \cite{ML2}.

However in general, physical systems are dynamical, i.e. they evolve in time. So a state (respectively an observable) can exhibit changes brought about by the passage of time.
\textit{With this aim in mind, we have in this paper defined and examined quantum maps which are able to describe dynamical processes within this same formalism.}

On the quantum level, we have presented two different ways of describing quantum maps, although the one is only directly available for semifinite algebras. Aside from that aspect, the starting point of both approaches is the same. Since an algebra of observables for a typical system in statistical mechanics or quantum field theory is usually given by a type III factor, a formalism suited to such algebras like the crossed product approach, should be used to get a suitable framework for non-commutative integration theory, 

The first approach we presented used the DDdP strategy, i.e. in the case where $\mathfrak{M}$ is semifinite, the selection of appropriate quantum measurable elements in $\widetilde{\mathfrak{M}}$ by means of generalized singular values. The second approach, which is based on Haagerup's crossed product strategy, employs certain relations based on the concept of fundamental functions, to identify those elements of $\widetilde{\cM}$ that may validly be regarded as elements of $L^\Phi(\mathfrak{M})$. A very important fact to take note of is that in the case where $\mathfrak{M}$ is a semifinite algebra equipped with a faithful normal semifinite trace, the two approaches yield the same structures! (See Section 2 of \cite{L}.) We emphasize that on ``the physical level'' it is difficult to say which method is better adapted to the theory of dynamical systems. As we showed, both approaches are well suited to the study of time evolution. However in view of the sheer power and elegance of the DDdP approach, this approach should be preferred in the tracial case. If on the other hand $\mathfrak{M}$ is a type III algebra, the breadth and scope of the crossed product approach is essential to success.

The second point which needs some elaboration, is that related to the role of DBC and CP. The starting point of the interpolation scheme demands a well defined map on a superspace of the interpolating spaces.
To get such a map, one firstly needs to extend the original dynamical map over the corresponding crossed product. To this end one at the very least needs the commutation of the considered map with the modular automorphism and also complete positivity. The former requirement is ensured by the DBC. It is worth pointing out that the so called property of KMS-symmetry/selfadjointness does not guarantee this property; see \cite{cip1}, \cite{Cip}. As far as the DBC property is concerned, we note that the ``physical'' meaning of DBC is that the considered dynamical systems are midway between equilibrium and full non-equilibrium time evolution (cf \cite{Maj1}). In other words, this condition excludes, to some extent, excessively irregular dynamical maps -- a restriction which is quite acceptable.

It was indicated at the beginning of the preceding section that the analysis of Markov maps has a rather long history; see also \cite{AF}. However, it was not until the appearance of Haagerup's pioneering paper \cite{Haage1} on the general theory of $L^p$-spaces (general means that type III factors were included) when it was recognized that this theory provides a very convenient framework for a study of both Markov maps and interpolation, see \cite{Kos}, \cite{MTerp}, \cite{GL}, \cite{Skalski}. But, the important point to note here is that this general theory of $L^p$-spaces ensures the existence of well defined canonical embeddings of the original ($\sigma$-finite) algebra $\mathfrak{M}$ into $L^p(\mathfrak{M})$-spaces.

Let a $\sigma$-finite von Neumann  algebra $\mathfrak{ M}$ be given in the standard form
$$ ( \ \mathfrak{M}, \ \cH, \  J, \ \cP, \ \Omega \ ). $$

Let $T: \mathfrak{M} \to \mathfrak{M}$  be a Markov map satisfying DBC. So neither complete positivity, nor any sort of selfadjointness with respect to modular structures is assumed. In this case, see subsection \ref{DBCo} and/or \cite{Maj1}, there is a one-to-one correspondence between the given map $T$ and the map $\widehat{T}: \cH \to \cH$ defined by
\begin{equation}
\widehat{T}x \Omega = T(x) \Omega, \quad x \in \mathfrak{M}.
\end{equation}
But as $$ (\ \mathfrak{M}, \ L^2(\mathfrak{M}), \ ^*, \ L^2_+(\mathfrak{M}), \  h_{\omega}^{\frac{1}{2}} \ )$$
is also a standard form for $\mathfrak{M}$, and standard forms for a fixed von Neumann algebra are universal, we infer that there exists (in one-to-one correspondence) an $L^2(\mathfrak{M})$-counterpart of $\widehat{T}$,
which can be identified with the map constructed on ${L^2(\mathfrak{M})}$ from the extension of $\tT$ to $(L^1+L^\infty)(\cM, \tau)$. This not only implies that $T$ can be recovered from the information encoded in the map induced on ${L^2(\mathfrak{M})}$, but also that when in standard form, the underlying Hilbert space is somehow intrinsic to the process of constructing quantum Orlicz spaces. 

We pause to observe that the complete positivity of $T$ can be described at the Hilbert space level. This point was made at the start of section 5 of \cite{LMM}, but for maps $T$ satisfying Detailed Balance II (see \cite{MS}) not DBC. For the sake of the reader, we show how that discussion may be adapted to the present context. (In the discussion in \cite{LMM}, the roles of the vectors $\widehat{T}(\Omega)$ and $\widehat{T}^*(\Omega)$ weren't described clearly enough. The discussion below alleviates that problem to some extent.) 

We let $\mathcal{P}_n$ denote the natural cone for $(\mathfrak{M}
\otimes B(\mathbb{C}^n), \omega \otimes \omega_0 )$ where $\omega_0$
is a faithful state on $B(\mathbb{C}^n)$ (one can for example take $\omega_0$ to be
$\frac{1}{n} Tr$). For
the same algebra, $\Delta_n = \Delta \otimes \Delta_0$ and $J_n=J\otimes J_0$
are respectively the modular operator and modular conjugation
for $M_n(\mathfrak{M})$, defined in terms of the vector
$\Omega_n = \Omega \otimes \Omega_0$ (ie. in
terms of the state $\omega \otimes \omega_0$).

We know from \cite[Theorem 4.12]{Maj1}, that $T$ is unital positive map on $\mathfrak{M}$ satisfying DBC, if and only if the prescription $\widehat{T}(a\Omega)=T(a)\Omega$, ($a\in \mathfrak{M}$), induces a bounded map $\widehat{T}$ on $\mathcal{H}$ which 
\begin{itemize}
\item commutes strongly with $\Delta$;
\item is $J$ selfadjoint in the sense that $J\widehat{T}J=\widehat{T}^*$;
\item satisfies $\widehat{T}(\Omega)=\Omega$;
\item and for which we have that $\widehat{T}(\mathcal{P}) \subset \mathcal{P}$.
\end{itemize}

Now suppose that $T$ and $\widehat{T}$ are given as above. Then for any $n\in \mathbb{N}$, $T$ induces a map $T\otimes \I$ on $M_n(\mathfrak{M})=\mathfrak{M}\otimes M_n(\mathbb{C})$ which canonically corresponds to the related bounded map $\widehat{T}\otimes \I$ on $\mathcal{H}^{(n)}$. From the fact that $\widehat{T}$ satisfies the first three bullets above, it is a simple matter to verify that similarly $\widehat{T}\otimes \I$ satisfies similar conditions with respect to $\Delta_n = \Delta \otimes \Delta_0$, $J_n=J\otimes J_0$, and $\Omega_n = \Omega \otimes \Omega_0$. Hence at the Hilbert space level, the one property which will determine if $\widehat{T}\otimes \I$ corresponds to a unital positive map $T\otimes \I$ on $M_n(\mathfrak{M})$ (satisfying DBC), is its action on the natural cone $\mathcal{P}_n = \overline{{\Delta}^{1/4}_n\{[a_{ij}]\Omega_n :
[a_{ij}] \in M_n(\mathfrak{M})^+\}}$ of $\mathcal{H}^{(n)}$. So given a map $T$ on $\mathfrak{M}$, we can say that $T$ is a unital completely positive map if and only if the prescription $\widehat{T}(a\Omega)=T(a)\Omega$, ($a\in \mathfrak{M}$), induces a bounded map $\widehat{T}$ on $\mathcal{H}$ which 
\begin{itemize}
\item commutes strongly with $\Delta$;
\item is $J$ selfadjoint in the sense that $J\widehat{T}J=\widehat{T}^*$;
\item satisfies $\widehat{T}(\Omega)=\Omega$;
\item and for which we have that $(\widehat{T}\otimes\I)(\mathcal{P}_n) \subset \mathcal{P}_n$ for all $n\in \mathbb{N}$.
\end{itemize}

In the discussion preceding Corollary \ref{orlcor}, we noted that with $T$ as above, one may construct densely defined maps $T^{(p)}$ on each $L^p(\mathfrak{M})$ ($1\leq p <\infty$), and then extend these by continuity to get an action on the given $L^p$ space. Using this approach, Goldstein and Lindsay developed a theory of Markov semigroups and Dirichlet forms specialised to $L^p$-spaces. (See \cite{GL}, \cite{GL2}.) The advantage of their approach is that they do not require the ambient maps to be CP. The disadvantage is that their approach only seems to work for $L^p$-spaces, whereas ours includes all the Orlicz spaces covered by Theorem \ref{orlthm} -- in particular the Orlicz space $L^{\cosh-1}$. In their analysis, Goldstein and Lindsay did also restrict attention to KMS-symmetric quantum maps, whereas no such assumption is made in the present work. Unpublished joint work of Goldstein and Lindsay with Skalski \cite{GLS} suggests that this assumption was one of convenience rather than necessity. We will nevertheless comment on this restriction.

\begin{enumerate}
\item 
As far as KMS-symmetry/selfadjointness is concerned, we firstly note that modular dynamics, the most serious candidate for 
equilibrium dynamics, does not satisfy this condition. To see this 
let $\mathfrak{M}$ be a von Neumann algebra, $\omega$ a faithful normal state on it. A Markov map
$\Phi: \mathfrak{M} \to \mathfrak{M}$ is said to KMS-self-adjoint if
\begin{equation}
\omega(\Phi(x) \sigma_{-\frac{i}{2}}(y)) = \omega(\sigma_{\frac{i}{2}}(x) \Phi(y))
\end{equation}
where $\sigma_t$ stands for the modular dynamics, $x,y$ are analytic elements in $\mathfrak{M}$.
We see at once that $\sigma_t$ is not KMS-self-adjoint.
\item As a second example illustrating difficulties, we consider a uniformly continuous semigroup 
$T_t: \mathfrak{M} \to \mathfrak{M}$, where $\mathfrak{M}$ is a von Neumann algebra acting on a Hilbert space $\cH$. Let $L$ be its infinitesimal generator.
Consider the following symmetric embedding (see \cite{cip1}):
$$i_0: \mathfrak{M} \to \cH$$
$$i_0(x) = \Delta^{\frac{1}{4}} x \Omega, \quad x \in \mathfrak{M},$$
where $\Omega \in \cH$ is such that $\omega(\cdot) = (\Omega, \cdot \Omega).$
Define the operator $H$ on $\cH$ by
\begin{equation}
H \Delta^{\frac{1}{4}}x \Omega = \Delta^{\frac{1}{4}}L(x) \Omega, \quad x \in \mathfrak{M}.
\end{equation}
Then, the operator $H$ associated to $L$ is selfadjoint if and only if certain complicated relations
between dissipative and Hamiltonian parts of $L$ are satisfied, for details see \cite{park}.
These additional relations have no direct ``physical'' meaning and should be interpreted cautiously.
\end{enumerate}

Finally, it is worth pointing out a philosophical difference between the two approaches. In the above approach one defines the maps on $L^p$ from ``inside'' by making use of the embedding of $\mathfrak{M}$ into $L^p(\mathfrak{M})$. By contrast our approach was to ``enlarge'' the operator to a big superspace which is large enough to allow for a certain degree of interpolation. We believe this to be a useful advantage.

\end{document}